\documentclass[11pt,a4paper]{article}

\usepackage[utf8]{inputenc}
\usepackage[T1]{fontenc}
\usepackage{lmodern}
\usepackage[english]{babel}
\usepackage[margin=2.7cm]{geometry}
\usepackage{amsmath,amssymb}
\usepackage{amsthm}
\usepackage{booktabs}
\usepackage{multirow}
\usepackage{longtable}
\usepackage{graphicx}
\usepackage{algorithm}
\usepackage{algpseudocode}
\usepackage{tikz}
\usetikzlibrary{arrows.meta,positioning,calc}
\usepackage{url}
\usepackage[hidelinks]{hyperref}
\usepackage{microtype}

\newtheorem{proposition}{Proposition}

\newcommand{\dpr}{D.P.R.\ 361/1957}
\newcommand{\arth}{Art.~83(1)(h)}
\newcommand{\arti}{Art.~83(1)(i)}

\title{One Vote, Several Parliaments:\\
An Empirical Analysis of the Algorithmic Ambiguity\\
of the Italian Electoral Law on the 2022 General Election Data}

\author{Paolo Coppola\\
\small Department of Mathematics, Computer Science and Physics, University of Udine\\
\small \texttt{paolo.coppola@uniud.it}}

\date{July 2026}

\begin{document}

\maketitle

\begin{abstract}
Crafa's algorithmic analysis of the Italian electoral law (\dpr, as amended by Law 165/2017, the so-called \emph{Rosatellum}) showed that the statutory text describing the territorial distribution of proportional seats (\arth) admits at least three different algorithmic interpretations, which may assign seats to different territories and hence elect different people from the same votes. We test that conclusion empirically: we implement the full seat-allocation pipeline of the law (Arts.~77, 83, 83-\emph{bis}, 84 and 85) and run the three interpretations on the complete open data of the Italian general election of 25 September 2022 (Chamber of Deputies). The implementation reproduces the official national apportionment exactly from the raw municipal data, matches by name 389 of the 391 seats it models (99.5\%), and agrees step by step with the official minutes of the National Central Electoral Office; the two residual disagreements fall on two of the four seats that the Chamber's own Committee on Elections placed under formal investigation in July 2025: in one case because the official record itself reports the decisive percentage tallies in two inconsistent ways, in the other because the decisive decimal comparison lies within the uncertainty of the verified vote counts. On these validated data, the sequential interpretation (Algorithm~A) is order-dependent: reversing the processing order of the constituencies replaces 6 deputies with 6 others, and 1{,}000 random orders produce 560 distinct outcomes; a near-exhaustive sweep of 15 million orders closes the set of order-contingent deputies at 140 persons, while 764 of the 1{,}139 candidates of the admitted lists (67\%) are elected under no order at all. In 29\% of the sampled orders A also strands one or two seats it cannot assign by any rule stated in the text, and in a further 9\% it fills the Chamber with a different party composition. The interpretation applied in electoral practice (Algorithm~C) is order-independent, provably so in the absence of ties, but differs from A by 8 deputies. Under the Mattarella-style interpretation (Algorithm~B) the documented order leaves two seats that cannot be assigned by any rule stated in the text, as historically occurred in 1994, 1996 and 2001, and 46\% of sampled orders strand one to three seats. Under the executions documented in practice the statutory compensation preserves the national seat totals of every coalition and list, and a Monte Carlo analysis shows that the named differences between interpretations are robust to input noise well beyond the residual uncertainty of the data: the ambiguity changes \emph{which persons} are elected and \emph{in which territories}, and leaves party strength untouched only where the text's procedure completes. On real data, this confirms the central claim of Crafa's analysis.
\end{abstract}

\section{Introduction}
\label{sec:intro}

Electoral laws are algorithms written in natural language. They take as input the votes cast by citizens and produce as output the composition of a parliament. When the natural-language description of this computation is ambiguous, the same electoral input may lead to different outputs, depending on which of the admissible readings of the text is executed. This is not a hypothetical concern: Crafa~\cite{crafa2024,crafa2023} showed that the text of the Italian electoral law currently in force, the consolidated law for the election of the Chamber of Deputies (\dpr) as amended by Law 165/2017 (``Rosatellum''), admits several distinct algorithmic interpretations of the procedure that distributes proportional seats among the electoral constituencies (\emph{circoscrizioni}), and that these interpretations can produce different territorial seat assignments. Crafa's analysis identifies three algorithms, here called A, B and C (Section~\ref{sec:crafa}), all \emph{algorithmically admissible} readings of the statutory text: a notion we define and delimit in Section~\ref{sec:legal-scope}, deliberately weaker than a claim of legal admissibility. The analysis further shows that the algorithm apparently applied in actual electoral practice (Algorithm~C) is not the most immediate sequential reading of the text, and that the statutory procedure for transferring surplus seats between constituencies is itself underspecified in rare but possible scenarios.

Crafa's work is an analysis of the legal text, illustrated on small abstract examples. It leaves open a natural empirical question: \emph{on the real votes of a real election, how large is the effect of the ambiguity?} Do the different admissible algorithms elect different parliaments, and if so, how many seats move, between which territories, which individual candidates gain or lose a seat in the Chamber, and with what consequences for the party composition?

This paper answers that question for the Italian general election of 25 September 2022 (Chamber of Deputies). Our contributions are the following.

\begin{enumerate}
\item We provide a complete, faithful, open implementation of the seat-allocation pipeline of the \dpr{} as in force for the 2022 election: Art.~77 (constituency-level tallies, including the statutory pro-rata reallocation of candidate-only votes), Art.~83 (national apportionment with thresholds, and territorial distribution in the three interpretations A/B/C, each followed by the statutory compensation procedure), Art.~83-\emph{bis} (distribution to multi-member colleges), and Arts.~84--85 (proclamation of elected candidates, resolution of multiple candidacies, and the fallback chain for exhausted lists). The implementation works directly on the raw open data of the Ministry of the Interior at municipality level (Section~\ref{sec:data}).
\item We validate the implementation against the official outcome at four increasingly demanding levels (Section~\ref{sec:validation}): exact reproduction of the official national apportionment computed from the raw data; agreement of national list totals within $\pm0.06\%$, a deviation attributable to the input file (the published open data are near-definitive, not the certified counts) rather than to the algorithm; a name-by-name comparison with the official proclamations obtained from the open data of the Chamber of Deputies, matching 389 of 391 modelled seats (99.5\%), the two residual disagreements falling on two of the four seats that the Chamber's own Committee on Elections placed under formal investigation in 2025 (seats on which the official record is internally inconsistent or decided within the uncertainty of the verified counts, so that the comparison target itself, not the simulation, is unsettled there); and a step-by-step comparison with the official minutes of the National Central Electoral Office, transcribed from the published scans (OCR followed by manual proof-reading of every value used in the comparison against the page images), in which the intermediate decimal parts, percentage tallies and recovery chains coincide with ours to within the residual noise of the input data. The name-level comparison is strict enough to expose single-sentence misreadings of the statute, and we document how it did so.
\item We quantify on complete real data the practical effect of the algorithmic ambiguity, \emph{down to the identity of the elected deputies} (Sections~\ref{sec:results} and~\ref{sec:names}). Algorithm~A is order-dependent: reversing the processing order of the constituencies replaces 6 deputies with 6 others and moves one more to a different college, and a systematic study over 1{,}000 uniformly random constituency orders (a sample of the $27! \approx 1.1\times10^{28}$ possible orders, whose exhaustive enumeration is out of reach) yields 560 distinct outcomes, with up to 13 deputies entering and up to 14 leaving the Chamber relative to practice. The sweep also shows that the sequential reading is not merely order-sensitive but, under 29\% of the sampled orders, \emph{incomplete}: one or two residual seats find no unit entitled to receive them and cannot be assigned by any rule stated in the text, and in a further 9\% of orders the Chamber is filled but with a different party composition (Section~\ref{sec:results}). Read per person, the 1{,}000 sampled orders identify 458 distinct individuals (339 elected under every sampled order, 119 order-contingent), each listed with their election frequency in Appendix~\ref{app:perm}; a near-exhaustive extension to 15 million orders --- feasible because the order enters only the fast Phase~3 of the pipeline --- closes this set at 469 individuals, of which 140 are order-contingent, so that of the 1{,}139 candidates of the lists admitted to the distribution, 764 (67.1\%) are elected under no order at all. Algorithm~C is order-independent: provably so in the absence of ties (Proposition~\ref{prop:C}; in case of ties the statute prescribes a draw by lot, so order-invariance would then hold up to the outcome of the draws), and empirically on the 2022 data. Since the outcome of A depends on the processing order, its comparison with C fixes the reference order, the one in which the worksheets of the National Central Electoral Office tabulate the 27 constituencies (Section~\ref{sec:results}): under that order, A and C differ by 8 deputies (4 replaced by 4). Algorithm~B fails to assign 2 seats by the rules stated in the text under the documented order; run under the same 1{,}000 sampled orders, it strands one to three seats in 45.7\% of them and reproduces the official party totals in only 15.7\%. We name every affected deputy, and for each change we publish the complete chain of arithmetic (constituency quotient, integer and decimal parts, exclusion events, compensation transfers, college placement, list ranking, multi-candidacy resolution) so that every claim can be re-verified with pencil and paper from the tables in this article, without consulting the code or the datasets.
\item We assess the robustness of the named results against the residual uncertainty of the input data, i.e.\ the possibility that the certified tallies differ from the near-definitive open data we compute from (Section~\ref{sec:names-robust}): under Monte Carlo perturbations of every list tally, at and well beyond the $\pm0.06\%$ level at which the open data deviate from the official tallies, the full set of named differences between interpretations persists in every draw, because each difference is anchored to integer-valued saturation events or to decimal-part races whose margins exceed the attainable perturbation by one to two orders of magnitude. Under the executions documented in practice, the statutory compensation renders national party totals invariant: there, the ambiguity displaces \emph{people and territories}, not party strength. Under Article~67 of the Constitution, which makes every member of Parliament an unbound representative of the Nation, the people are the constitutionally operative part. Under arbitrary processing orders, as Section~\ref{sec:results} shows, even this invariance can fail.
\end{enumerate}

Beyond the specific case, the paper gives concrete empirical form to a conceptual point in computational law made by Crafa's analysis, and by the law-as-code tradition before it (Section~\ref{sec:related}): formalising a statute does not merely implement it: it can reveal that the enacted text fails to specify a unique decision function. The ambiguity documented here is not lexical (no word of \arth{} is obscure, and no single sentence is contested) but \emph{operational}: the text prescribes elementary operations without fixing their schedule, and different admissible schedules of the same words elect different parliaments. Ambiguity of this kind is largely invisible to doctrinal reading, precisely because each clause taken alone is clear; it becomes visible when the text is executed. What this paper adds is a \emph{measure} of that point: executed on the complete votes of a real election, the admissible readings elect measurably different parliaments, down to named individuals, with margins that survive the residual uncertainty of the data.

All code, derived datasets and logs are also available for inspection and re-execution (Section~\ref{sec:repro}).

\section{Related work}
\label{sec:related}

\paragraph{Law as code.} The observation that rendering a statute executable is itself a powerful form of legal analysis goes back at least to the formalisation of the British Nationality Act as a logic program by Sergot et al.~\cite{sergot1986}, which treated legislation as a specification whose translation into code forces latent imprecisions to the surface. The programme is active today: the Catala language of Merigoux et al.~\cite{merigoux2021catala} was designed for the systematic translation of statutory law into executable implementations, on the explicit motivation that legal prose ``leaves room for ambiguities'' that only formalisation exposes, and the same group compiled the French tax code into a verified artefact~\cite{merigoux2021mlang}. On the institutional side, the OECD's \emph{Rules as Code} programme~\cite{oecd2020} advocates publishing an official machine-consumable counterpart of natural-language rules, precisely because the status quo requires every implementer to re-interpret the text independently, the mechanism by which the ambiguity studied here produces divergent outcomes. From legal theory, Diver~\cite{diver2021} examines what is normatively at stake when rules are enforced by code, which executes deterministically and offers none of the interpretive latitude of legal texts; and recent work explores language models as an aid to statutory formalisation, again finding that making a provision executable deterministically surfaces defects that reading alone misses~\cite{yadamsuren2025}. Our study belongs to this tradition but inverts its usual direction: rather than proposing a formal notation into which the law should be re-drafted, we take the enacted text as it stands, execute \emph{all} of its algorithmically admissible readings on the complete data of a real national election, and measure how much the residual freedom matters.

\paragraph{The mathematics of apportionment.} A distinct literature studies seat-apportionment methods as mathematical objects: the classical treatment of Balinski and Young~\cite{balinski2001} and the modern monograph of Pukelsheim~\cite{pukelsheim2017} characterise divisor and quota methods, their paradoxes (Alabama, population, new-state) and the impossibility of satisfying all fairness axioms at once. That literature analyses the properties of \emph{well-defined} allocation functions and the choice among them. The phenomenon studied here is upstream of it and of a different nature: the statutory text fails to determine \emph{which} function is being computed at all. The Rosatellum's method (largest remainders with truncated quotients) is mathematically unremarkable; what is underdetermined is the operational schedule of an exclusion clause across constituencies, a purely textual degree of freedom that the apportionment literature, starting from well-posed definitions, has no occasion to encounter.

\paragraph{Verified vote counting and election auditing.} A third strand applies formal methods to the \emph{counting} of votes: logical analyses of counting schemes~\cite{beckert2013}, vote counting as mathematical proof with machine-checkable certificates~\cite{pattinson2015}, and mechanised verification of STV algorithms~\cite{ghale2018}; complementarily, risk-limiting audits give statistical guarantees that a reported outcome matches the ballots~\cite{stark2008}. The closest single precedent to our findings is the analysis by Conway et al.~\cite{conway2017} of the New South Wales electronic vote count: the enabling legislation was ambiguous on which transfer counts as the ``last'' when several candidates are elected at the same count, the official specification's prose contradicted its own pseudocode, and the certified software implemented the defective reading, demonstrably changing a real council outcome. That episode concerns one clause of one local STV count; the present paper documents the same phenomenon (statutory text underdetermining the counting algorithm, with outcome-changing effect) systematically, across the entire proportional tier of a national general election, and traces every affected seat to a named person. Methodologically, our step-level validation against the OCR-transcribed minutes of the National Central Electoral Office (Section~\ref{sec:validation}) is a form of post-hoc audit of the official computation, in the spirit of this tradition; the four seats that the Chamber's Committee on Elections placed under investigation in 2025 are its institutional counterpart.

\paragraph{The Rosatellum in political science.} Finally, the Italian electoral system of 2017 has an established political-science literature: its mechanics and its effects on strategic coordination and disproportionality~\cite{chiaramonte2018ips}, the reform trajectory that produced it~\cite{massetti2019}, and the party-system consequences of the 2022 election run under it~\cite{chiaramonte2022sesp,chiaramonte2023}. That literature evaluates the system's political \emph{outputs}, taking the seat-allocation procedure as a well-defined function of the votes. Whether the enacted text actually determines that function is the question raised by Crafa~\cite{crafa2024}, and answered empirically, at the level of individual deputies, here.

\medskip
To the best of our knowledge, this paper provides the first empirical quantification of a statutory ambiguity carried out at the level of individually named elected representatives on the complete data of a real national election; we found no prior work, in any of the four strands above, that does so.

\section{The seat-allocation procedure of the Rosatellum}
\label{sec:law}

We summarise the procedure of the \dpr{}~\cite{dpr361} as in force on 25 September 2022, restricted to the Chamber of Deputies. Of the 400 deputies, 147 are elected in single-member districts (\emph{collegi uninominali}) by plurality, 245 are elected proportionally in 49 multi-member colleges (\emph{collegi plurinominali}) grouped into 27 constituencies (\emph{circoscrizioni}), and 8 are elected in the overseas constituency. The Aosta Valley forms a separate single-member constituency. The proportional allocation proceeds top-down through three territorial levels (national, constituency, multi-member college), each level fixing the seat totals that the level below must respect. Throughout the procedure, electoral quotients are \emph{truncated} integer divisions (``no account is taken of the fractional part of the quotient''), and seats left over after integer quotients are assigned by largest remainders (\emph{parti decimali}, decimal parts of the quotients of attribution). Table~\ref{tab:glossary} collects the territorial vocabulary used throughout the paper, and Figure~\ref{fig:pipeline} the pipeline as a whole.

\begin{table}[t]
\centering\small
\caption{The recurring vocabulary of the allocation procedure. Italian terms are given because they appear in the statute and in the official records cited.}
\label{tab:glossary}
\begin{tabular}{p{4.4cm}cp{6.2cm}}
\toprule
Term (Italian) & 2022 & Meaning \\
\midrule
constituency (\emph{circoscrizione}) & 27 & first-level territorial unit of the proportional allocation (Aosta Valley excluded) \\
multi-member college (\emph{collegio plurinominale}) & 49 & subdivision of a constituency; closed lists of 2--4 candidates compete here \\
single-member district (\emph{collegio uninominale}) & 147 & plurality seat; 146 modelled (Aosta Valley excluded) \\
unit & 5 & coalition or single list admitted to the national apportionment \\
tally (\emph{cifra elettorale}) & --- & vote total of a list or coalition at a given territorial level, after the pro-rata reallocation of Art.~77(1)(c) \\
electoral quotient (\emph{quoziente}) & --- & truncated integer ratio of total tallies to seats at a given level \\
decimal part (\emph{parte decimale}) & --- & fractional part of a unit's quotient of attribution; ranks the claims to residual seats \\
surplus / deficit unit (\emph{eccedentaria / deficitaria}) & --- & unit above/below its entitlement after a distribution; repaired by the statutory compensation \\
\bottomrule
\end{tabular}
\end{table}

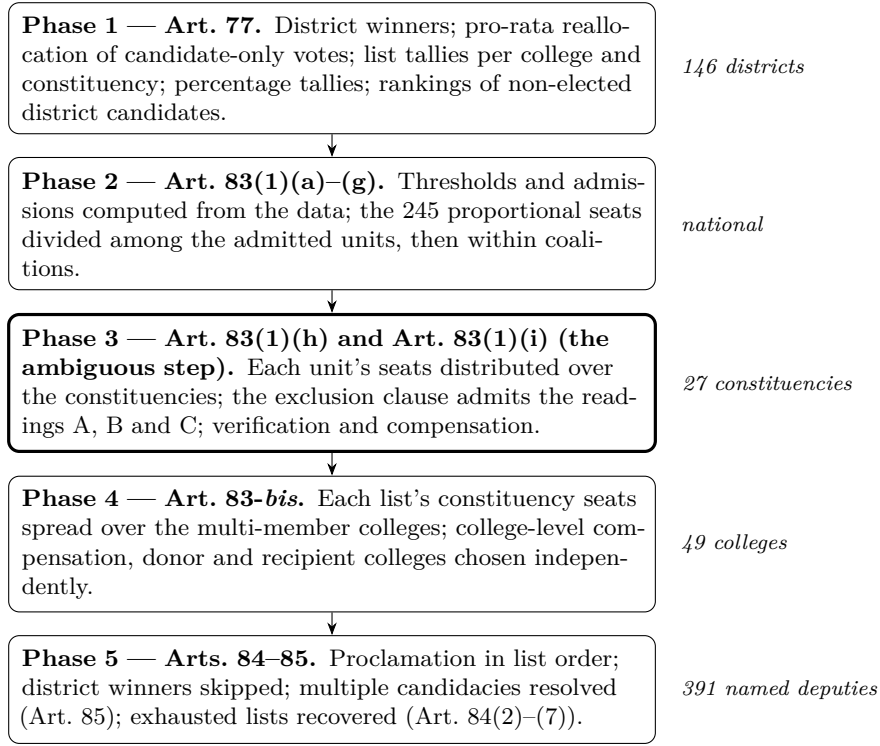
\begin{figure}[t]
\centering
\begin{tikzpicture}[
  phase/.style={draw, rounded corners, align=left, text width=8.2cm, inner sep=5pt, font=\footnotesize},
  lev/.style={font=\scriptsize\itshape, anchor=west},
  node distance=3mm]
\node[phase] (f1) {\textbf{Phase 1 --- Art.~77.} District winners; pro-rata reallocation of candidate-only votes; list tallies per college and constituency; percentage tallies; rankings of non-elected district candidates.};
\node[phase, below=of f1] (f2) {\textbf{Phase 2 --- Art.~83(1)(a)--(g).} Thresholds and admissions computed from the data; the 245 proportional seats divided among the admitted units, then within coalitions.};
\node[phase, below=of f2, very thick] (f3) {\textbf{Phase 3 --- \arth{} and \arti{} (the ambiguous step).} Each unit's seats distributed over the constituencies; the exclusion clause admits the readings A, B and C; verification and compensation.};
\node[phase, below=of f3] (f4) {\textbf{Phase 4 --- Art.~83-\emph{bis}.} Each list's constituency seats spread over the multi-member colleges; college-level compensation, donor and recipient colleges chosen independently.};
\node[phase, below=of f4] (f5) {\textbf{Phase 5 --- Arts.~84--85.} Proclamation in list order; district winners skipped; multiple candidacies resolved (Art.~85); exhausted lists recovered (Art.~84(2)--(7)).};
\node[lev, xshift=2mm] at (f1.east) {146 districts};
\node[lev, xshift=2mm] at (f2.east) {national};
\node[lev, xshift=2mm] at (f3.east) {27 constituencies};
\node[lev, xshift=2mm] at (f4.east) {49 colleges};
\node[lev, xshift=2mm] at (f5.east) {391 named deputies};
\draw[-{Stealth}] (f1) -- (f2);
\draw[-{Stealth}] (f2) -- (f3);
\draw[-{Stealth}] (f3) -- (f4);
\draw[-{Stealth}] (f4) -- (f5);
\end{tikzpicture}
\caption{The seat-allocation pipeline of the \dpr{} for the Chamber of Deputies (right: the territorial level at which each phase operates). Each phase fixes totals that the next must respect; every difference among the interpretations A, B and C originates in Phase~3 and propagates mechanically through Phases~4--5.}
\label{fig:pipeline}
\end{figure}

\paragraph{Art.~77 (constituency electoral offices).} For each single-member district: the candidate with the most votes is elected (ties are broken in favour of the youngest candidate). Votes cast for a district candidate \emph{without} marking any supporting list (``candidate-only'' votes) are redistributed pro rata among the lists of the candidate's coalition, in proportion to their list votes in that district, using integer arithmetic with largest remainders (Art.~77(1)(c)). The resulting list tallies (\emph{cifre elettorali}) are accumulated at the level of the multi-member college and of the constituency; percentage tallies (Art.~77(1)(e)) and rankings of non-elected district candidates (Art.~77(1)(h)) are also computed, both of which are needed later by Arts.~84--85.

\paragraph{Art.~83(1)(a)--(g) (national office: apportionment).} National list tallies are summed; a list counts towards its coalition's national tally only if it obtained at least 1\% of the national valid votes (or qualifies under the linguistic-minority rule). Admitted to the apportionment are: coalitions with at least 10\% nationally that contain at least one list above 3\%; single lists above 3\%; and linguistic-minority lists presented in a special-statute region with at least 20\% of the regional vote or at least two elected district candidates. The 245 proportional seats are divided among the admitted units by the national electoral quotient (truncated) and largest remainders (letter~(f)); within each coalition, seats are then divided among its above-threshold lists by the same method (letter~(g)).

\paragraph{\arth{} (territorial distribution --- the ambiguous step).} The national office then distributes the seats of each coalition or single list among the constituencies. For each constituency, the number of proportional seats is fixed in advance (seats of the constituency minus its single-member districts); the constituency electoral quotient is the truncated ratio of the sum of the constituency tallies of the admitted units to that number of seats; each unit receives the integer part of its quotient of attribution, and the remaining seats go to the units with the largest decimal parts, \emph{excluding} units that have already reached their national seat total from letter~(f). The law then prescribes a verification: if, summing over all constituencies, some unit has more seats than its national entitlement (\emph{eccedentaria}) and some other has fewer (\emph{deficitaria}), surplus seats are transferred: the surplus unit gives up seats in the constituencies where it obtained them with the \emph{smallest} decimal parts, and, with priority \emph{within the same constituency}, they are given to deficit units with the largest unused decimal parts. As Crafa showed, the highlighted exclusion clause makes the whole step ambiguous: the text does not specify how the per-constituency operations are interleaved, and different admissible interleavings yield different assignments (Section~\ref{sec:crafa}).

\paragraph{\arti{}.} The same distribution-plus-compensation scheme is repeated \emph{inside} each coalition, to distribute each coalition's constituency seats among its lists, with exclusions driven by the national per-list totals of letter~(g).

\paragraph{Art.~83-\emph{bis} (multi-member colleges).} Each constituency office distributes the constituency seats of each list among the multi-member colleges of the constituency. The college quotient is the truncated ratio of the sum of the college tallies of the lists to the college's predetermined seat count. The text says ``of all the lists'', but the official worksheets show that the sum is taken over the lists \emph{admitted to the apportionment} only (Section~\ref{sec:validation}), one more point on which the executed procedure specifies the text. Then: integer parts, largest decimal parts with exclusion of lists that have reached their constituency total, and a compensation procedure that differs pointedly from that of \arth{}: at this level the transferred seat may land in a college different from the one where it was given up, so the realised number of seats of a college may differ from its predetermined number.

\paragraph{Arts.~84--85 (proclamation).} List candidates are proclaimed in list order. A candidate elected in a single-member district is skipped in the proportional lists. A candidate winning in several multi-member colleges is proclaimed in the college where their list obtained the \emph{lowest} percentage tally (Art.~85). If a list exhausts its candidates in a college (\emph{incapienza}), Art.~84(2)--(7) prescribes a fallback chain: other colleges of the same constituency (by decreasing decimal parts), then the ranking of the list's non-elected district candidates in the constituency, then other constituencies, then their district rankings.

\section{Crafa's three algorithmic interpretations}
\label{sec:crafa}

Crafa~\cite{crafa2023,crafa2024} decomposes \arth{} into six elementary steps for each constituency $C_i$: (1) compute the number of seats $s_i$ to distribute in $C_i$; (2) compute the constituency electoral quotient; (3) compute the quotients of attribution of the $n$ admitted units; (4) assign to each unit the integer part of its quotient; (5) assign the remaining seats by decreasing decimal parts; (6) exclude from step (5) the units that have already been assigned their national seat total. The text specifies neither the order of the constituencies nor how steps (4)--(6) are interleaved across constituencies, and the exclusion clause of step (6) makes the result depend on both choices. Three interpretations arise. (In the published version of Crafa's analysis~\cite{crafa2024} they are named Algorithms~I, II and~III; we keep the A/B/C naming of the technical report~\cite{crafa2023}, with A${}={}$I, B${}={}$II, C${}={}$III.)

\begin{description}
\item[Algorithm A (sequential).] The most natural reading of the text: perform steps (1)--(6) completely in the first constituency, then in the second, and so on. The exclusion of step (6) is evaluated \emph{dynamically}: a unit that reaches its national total in an early constituency is excluded from residual seats in all later ones. The outcome therefore depends on the order in which the constituencies are processed.

\item[Algorithm B (Mattarella-style).] First perform steps (1)--(4) in all constituencies (all integer parts everywhere); then sweep the constituencies again, assigning residual seats by decimal parts with the \emph{dynamic} exclusion of step (6). Since the exclusion is applied while residual seats are being assigned, no unit can exceed its national total and no compensation is needed; but the outcome still depends on the constituency order and, crucially, some seats may remain unassigned when the units still entitled to seats have no usable remainders in the remaining constituencies. Law 277/1993 (the ``Mattarella law'') adopted precisely this scheme, fixing the order (by increasing population); the unassigned-seat defect actually materialised in the elections of 1994, 1996 and 2001~\cite{camera2019}.

\item[Algorithm C (operational practice).] First assign all integer parts in all constituencies; then compute the exclusion of step (6) \emph{once}, statically, excluding only the units that have reached their national total \emph{on integer parts alone}; then assign all residual seats in every constituency by decimal parts, ignoring further saturation. Units may end up above their national total; the statutory compensation procedure then repairs the totals. Crafa's analysis of the parliamentary documentation (the 2019 study of the Chamber's research service~\cite{camera2019}, the 2020 report of the Committee on Elections~\cite{giunta2020} and the official Electoral Manual~\cite{manuale2022}) indicates that this, and not A, is the procedure actually applied in the official tallying.
\end{description}

Figure~\ref{fig:abc} shows the three schedules side by side; Algorithms~\ref{alg:A}--\ref{alg:C} give the three interpretations in pseudocode, in the formulation we implemented. All three are followed by the statutory verification-and-compensation step of \arth{} (last sentences of the letter), which we implement identically for all variants (Section~\ref{sec:impl-h}).

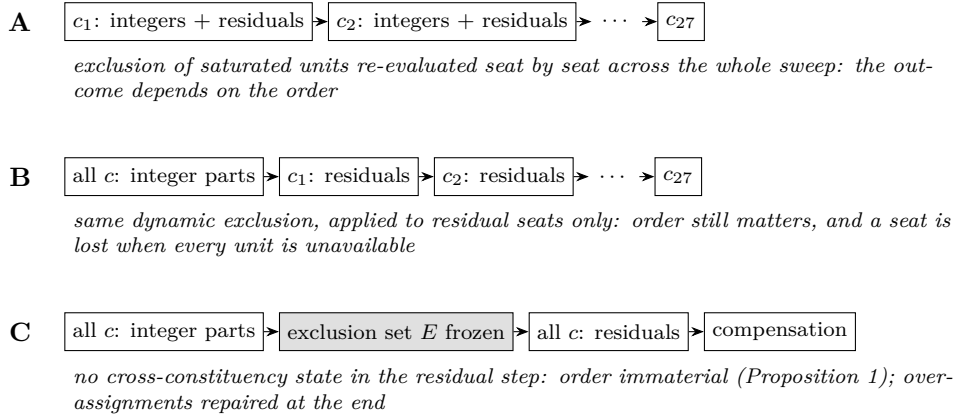
\begin{figure}[t]
\centering
\begin{tikzpicture}[font=\scriptsize,
  s/.style={draw, inner sep=3pt, minimum height=5mm, align=center},
  fz/.style={s, fill=black!12},
  rl/.style={font=\small\bfseries, anchor=east},
  nt/.style={font=\scriptsize\itshape, anchor=north west, text width=11.6cm},
  node distance=2mm]
%% row A
\node[rl] (A) at (0,0) {A};
\node[s, right=3mm of A] (a1) {$c_1$: integers $+$ residuals};
\node[s, right=of a1] (a2) {$c_2$: integers $+$ residuals};
\node[right=of a2] (a3) {$\cdots$};
\node[s, right=of a3] (a4) {$c_{27}$};
\draw[-{Stealth}] (a1) -- (a2); \draw[-{Stealth}] (a2) -- (a3); \draw[-{Stealth}] (a3) -- (a4);
\node[nt] at ($(a1.south west)+(0,-0.8mm)$) {exclusion of saturated units re-evaluated seat by seat across the whole sweep: the outcome depends on the order};
%% row B
\node[rl] (B) at (0,-2.05) {B};
\node[s, right=3mm of B] (b1) {all $c$: integer parts};
\node[s, right=of b1] (b2) {$c_1$: residuals};
\node[s, right=of b2] (b3) {$c_2$: residuals};
\node[right=of b3] (b4) {$\cdots$};
\node[s, right=of b4] (b5) {$c_{27}$};
\draw[-{Stealth}] (b1) -- (b2); \draw[-{Stealth}] (b2) -- (b3); \draw[-{Stealth}] (b3) -- (b4); \draw[-{Stealth}] (b4) -- (b5);
\node[nt] at ($(b1.south west)+(0,-0.8mm)$) {same dynamic exclusion, applied to residual seats only: order still matters, and a seat is lost when every unit is unavailable};
%% row C
\node[rl] (C) at (0,-4.1) {C};
\node[s, right=3mm of C] (c1) {all $c$: integer parts};
\node[fz, right=of c1] (c2) {exclusion set $E$ frozen};
\node[s, right=of c2] (c3) {all $c$: residuals};
\node[s, right=of c3] (c4) {compensation};
\draw[-{Stealth}] (c1) -- (c2); \draw[-{Stealth}] (c2) -- (c3); \draw[-{Stealth}] (c3) -- (c4);
\node[nt] at ($(c1.south west)+(0,-0.8mm)$) {no cross-constituency state in the residual step: order immaterial (Proposition~\ref{prop:C}); over-assignments repaired at the end};
\end{tikzpicture}
\caption{Where the three readings of \arth{} differ: the scheduling of the same elementary operations (integer parts, residual seats by decimal parts, exclusion of units that have reached their national entitlement) across the constituencies $c_1,\dots,c_{27}$.}
\label{fig:abc}
\end{figure}

\subsection{Algorithmic admissibility versus legal interpretation}
\label{sec:legal-scope}

A remark on the scope of the claim is in order, since ``interpretation'' has a narrower meaning in this paper than in legal doctrine. We call a procedure an \emph{algorithmically admissible reading} of a statutory provision if executing it performs exactly the operations the text prescribes, in an order consistent with the text, resolving only those choices that the text leaves open. This is a property of the text considered as the description of a computation. It is deliberately weaker than legal admissibility, which is decided by the canons of interpretation and weighs, beyond the literal wording: \emph{systematic} arguments (coherence with the rest of the statute; for instance, the verification-and-compensation step of \arth{} presupposes that surplus assignments can arise, which is natural under C, possible under A, and vacuous under B); \emph{historical} arguments (the continuity of C with the practice developed under the Mattarella law and with the intent of Law 165/2017 as reconstructed in the parliamentary documentation~\cite[p.~70]{camera2019}); \emph{administrative practice} (the operational manuals and the actual conduct of the tallying~\cite[p.~50]{manuale2022}\cite[Table F2, p.~95]{giunta2020}); and \emph{constitutional conformity} (a reading under which seats remain unassigned, as B, or under which the outcome depends on an arbitrary processing order, as A, would sit uneasily with the constitutional framework of the electorate's equal vote and the Chamber's fixed size). A full legal analysis might well conclude that C is the only \emph{legally} admissible interpretation precisely on the strength of these canons. Our point, and Crafa's, is different and narrower: the enacted text alone does not determine the computation, so the selection of C is currently performed by materials external to the statute (manuals, practice, tradition), not by the statute itself. Throughout the paper, phrases such as ``interpretation A'' and ``admissible readings'' are to be understood in the algorithmic sense; we make no claim that a returning office confronted with the text would be legally entitled, all things considered, to execute A or B.

\begin{algorithm}[t]
\caption{Interpretation A (sequential, dynamic exclusion)}
\label{alg:A}
\begin{algorithmic}[1]
\State $\mathit{tot}[u] \gets 0$ for every admitted unit $u$; let $T[u]$ be $u$'s national entitlement (letter (f))
\For{each constituency $c$ \textbf{in the chosen order}}
  \State compute quotient $q_c$, integer parts $b_c[u]$ and decimal parts $d_c[u]$
  \State $\mathit{seats}[c,u] \mathrel{+}= b_c[u]$;\quad $\mathit{tot}[u] \mathrel{+}= b_c[u]$
  \For{each residual seat of $c$, scanning units by decreasing $d_c[u]$}
    \If{$\mathit{tot}[u] < T[u]$} assign the seat to $u$ in $c$; $\mathit{tot}[u] \mathrel{+}= 1$ \EndIf
  \EndFor
\EndFor
\State run the statutory verification and compensation
\end{algorithmic}
\end{algorithm}

\begin{algorithm}[t]
\caption{Interpretation B (Mattarella-style, dynamic exclusion after all integer parts)}
\label{alg:B}
\begin{algorithmic}[1]
\For{each constituency $c$} compute $q_c$, $b_c$, $d_c$; assign integer parts; update $\mathit{tot}$ \EndFor
\For{each constituency $c$ \textbf{in the chosen order}}
  \For{each residual seat of $c$, scanning units by decreasing $d_c[u]$}
    \If{$\mathit{tot}[u] < T[u]$} assign the seat to $u$ in $c$; $\mathit{tot}[u] \mathrel{+}= 1$ \EndIf
  \EndFor
  \State \textit{(a residual seat may remain unassigned if all units are saturated)}
\EndFor
\State run the statutory verification and compensation
\end{algorithmic}
\end{algorithm}

\begin{algorithm}[t]
\caption{Interpretation C (operational manuals, static exclusion)}
\label{alg:C}
\begin{algorithmic}[1]
\For{each constituency $c$} compute $q_c$, $b_c$, $d_c$; assign integer parts; update $\mathit{tot}$ \EndFor
\State $E \gets \{\, u : \mathit{tot}[u] \ge T[u] \,\}$ \Comment{exclusion computed once, on integer parts only}
\For{each constituency $c$}
  \For{each residual seat of $c$, scanning units by decreasing $d_c[u]$}
    \If{$u \notin E$} assign the seat to $u$ in $c$; $\mathit{tot}[u] \mathrel{+}= 1$ \EndIf
  \EndFor
\EndFor
\State run the statutory verification and compensation \Comment{surpluses are expected here}
\end{algorithmic}
\end{algorithm}

\section{Data}
\label{sec:data}

All inputs are public. Table~\ref{tab:data} summarises them.

\begin{table}[t]
\centering\small
\caption{Data sources used by the simulation.}
\label{tab:data}
\begin{tabular}{p{5.2cm}p{9.3cm}}
\toprule
Dataset & Content and provenance \\
\midrule
Municipal-level results \newline ({\footnotesize\path{Camera_Italia_LivComune.txt}}) & Votes for every list and every single-member candidate in every municipality (25\,MB, $\approx$ definitive data). Open data of the Ministry of the Interior (Eligendo platform)~\cite{eligendo}. \\
Candidate files (\texttt{...\_uni.csv}, \texttt{...\_pluri.csv}) & Single-member candidates and closed proportional lists (candidate names, birth dates, list, position on the list). Ministry of the Interior~\cite{eligendo}. \\
Proportional seats per constituency & 245 seats over 27 constituencies (Aosta Valley excluded), per the apportionment based on the 2011 census in force in 2022. \\
Seats per multi-member college & 49 colleges, from the college design of D.Lgs.\ 177/2020~\cite{dlgs177}; reconstructed and validated as described below. \\
Official proclamations & Elected deputies with college and list, extracted via SPARQL from the linked open data of the Chamber of Deputies (\texttt{dati.camera.it})~\cite{daticamera}, mandates starting October 2022: 400 deputies $=$ 147 single-member $+$ 245 proportional $+$ 8 overseas. \\
Statutory text & \dpr{} Arts.\ 67--88 in the 2022 wording, from the official Electoral Manual~\cite{manuale2022}; \arth{} also reproduced in full in the appendix of~\cite{crafa2023}. \\
Official minutes & Minutes of the National Central Electoral Office for the 2022 election (\emph{verbale UECN}, Modello n.~89 E.P.), parts 3, 4 and 8, published as scanned images~\cite{uecn2022}; transcribed by OCR and used for the step-level validation of Section~\ref{sec:validation}. \\
\bottomrule
\end{tabular}
\end{table}

Three properties of the raw data required careful treatment; we document them because they materially affect any attempt to reproduce the official computation.

\paragraph{Candidate-only votes are absent from the list vote counts.} In the municipal file, the field \texttt{VOTILISTA} contains only the votes in which the voter marked the list. The 1{,}027{,}135 votes (3.66\% of the total) cast for a district candidate without marking a list are recoverable only as the difference between the candidate's votes and the sum of their supporting lists' votes, district by district. Ignoring them understates every list tally, non-uniformly, by up to several percent; Art.~77(1)(c) requires them to be redistributed pro rata (Section~\ref{sec:impl}).

\paragraph{Candidate votes are repeated per list-row.} The field \texttt{VOTICANDIDATO} is repeated on every list-row of a coalition candidate and must be de-duplicated before summation; after de-duplication, the reconstruction is exact for all 1{,}303 single-member candidacies.

\paragraph{Name conventions and encoding.} The candidate files store names in a single field, first name first (``\texttt{FIRSTNAME} \texttt{LASTNAME}''), whereas the votes file has separate \texttt{LASTNAME} and \texttt{FIRSTNAME} fields; person records are joined by an order-insensitive key (sorted name tokens plus birth date). The string \emph{S\"udtirol(er)} appears with a corrupted character in both college and list names and must be normalised consistently, on pain of losing the SVP seat in the Trentino-Alto Adige P01 college.

\paragraph{Seats per multi-member college.} The per-college seat counts fixed by D.Lgs.\ 177/2020 are not distributed as a single machine-readable table with the results data; we reconstructed them and checked the reconstruction for consistency in two ways: (i) the counts sum to the proportional entitlement of each constituency, and (ii) they are consistent with the statutory bound on list length (each list presents between 2 and 4 candidates, and no more than the college's seats), which the actual candidate lists saturate; all 49 colleges pass both checks. Neither check is, by itself, a strong independent validation; the reconstruction is however also validated \emph{end-to-end}: an error in any college's seat count would alter the college-level distribution and break the name-by-name agreement with the official proclamations reported in Section~\ref{sec:validation}, which instead matches all constituency totals exactly and 389 of 391 elected deputies individually.

\paragraph{Scope.} The Aosta Valley constituency (one single-member seat, no proportional seats) is absent from the municipal file, and the overseas constituency follows a different electoral system; the simulation therefore models $391 = 146 + 245$ of the 400 seats. Neither exclusion interacts with the proportional pipeline under study.

\section{Implementation}
\label{sec:impl}

The reference implementation (\texttt{rosatellum.py}, Python~3 with \texttt{pandas}) is a direct transcription of the statutory pipeline into five phases, mirroring the articles of the law. It is deliberately written as straight-line, integer-exact code: quotients are truncated integer divisions wherever the law says so, remainders are compared as exact fractions of the quotient, and every discretionary or random step mandated by the law (draws by lot, \emph{sorteggi}) is delegated to a seeded pseudo-random generator and logged. On the 2022 data no draw is ever actually needed, so the seed does not influence any reported result. Every non-trivial decision (threshold admissions, exclusions, compensation moves, multi-candidacy resolutions, exhausted-list fallbacks) is appended to a human-readable audit log, allowing step-by-step comparison with the official records.

\subsection{Phase 1 --- Art.~77}
\label{sec:impl-f1}

For each of the 146 single-member districts the implementation: (i) determines the winner by plurality, breaking ties in favour of the youngest candidate (Art.~77(1)(b); no tie occurs in 2022); (ii) computes each candidate's individual tally and each list's tally; (iii) redistributes candidate-only votes. For a candidate supported by a single list, candidate-only votes accrue to that list directly. For a coalition candidate supported by lists $\ell_1,\dots,\ell_k$ with district list votes $v_1,\dots,v_k$ and $v^{\ast}$ candidate-only votes, list $\ell_j$ receives $\lfloor v_j\,v^{\ast} / \sum_i v_i \rfloor$ votes, and the votes still missing are assigned one each to the lists with the largest remainders $v_j\,v^{\ast} \bmod \sum_i v_i$, integer arithmetic throughout, so that redistributed votes sum exactly to $v^{\ast}$. The phase finally accumulates list tallies per multi-member college and per constituency, percentage tallies per college (needed by Art.~85), and the per-list rankings of non-elected district candidates by decreasing percentage (Art.~77(1)(h), needed by Art.~84(3) and (6)).

\subsection{Phase 2 --- Art.~83(1)(a)--(g)}
\label{sec:impl-f2}

National tallies and all thresholds are \emph{computed} from the data, not assumed. On the 2022 data the computation gives: the lists \emph{Impegno Civico} (0.60\%) and \emph{Noi Moderati} (0.91\%) fall below the 1\% bar and are excluded from their coalitions' tallies, while \emph{+Europa} (above 1\%, below 3\%) counts towards the centre-left coalition tally without being admitted to seats; the SVP--PATT list qualifies under the linguistic-minority rule with 23.1\% of the regional vote and two elected district candidates. Five units are admitted to the apportionment: the centre-right coalition (12{,}048{,}826 votes, 42.88\% of the 28{,}096{,}790 national valid list votes), the centre-left coalition (7{,}170{,}619, 25.52\%), \emph{Movimento 5 Stelle} (4{,}336{,}660, 15.43\%), \emph{Azione--Italia Viva} (2{,}187{,}296, 7.78\%) and SVP--PATT (117{,}010, 0.42\%, via the minority rule). The 245 seats are divided by the truncated national quotient $\lfloor 25{,}860{,}411 / 245 \rfloor = 105{,}552$ and largest remainders (letter~(f)), then within each coalition among its above-threshold lists (letter~(g)). The results are reported in Section~\ref{sec:validation} (Table~\ref{tab:apportionment}). The national valid-vote total that serves as the denominator of the 1\% and 3\% thresholds is taken over the 27 mainland constituencies only; the Aosta Valley and overseas votes (Section~\ref{sec:data}) do not enter it. This treatment is confirmed end-to-end: the exact reproduction of the official apportionment at Level~1 (Section~\ref{sec:validation}) would be impossible had the threshold denominator been wrong, since the admitted set of units and the 245-seat division both depend on it.

\subsection{Phase 3 --- \arth{} and \arti{}, in variants A, B and C}
\label{sec:impl-h}

The territorial distribution is implemented once, as a single routine parameterised by the variant, and instantiated twice: at the level of coalitions/single lists over the 27 constituencies (\arth{}), and, within each coalition, at the level of its lists over the same constituencies (\arti{}). The three variants differ only in where the exclusion clause is evaluated, exactly as in Algorithms~\ref{alg:A}--\ref{alg:C}. Ties on decimal parts are broken by the larger national tally, as the law prescribes, and then by simulated draw (never reached in 2022). The constituency processing order is an explicit parameter, which is what allows the order-dependence experiments of Section~\ref{sec:results}.

All variants are followed by the same implementation of the statutory verification and compensation. While surplus units exist, the surplus unit with the most excess seats gives up one seat, chosen in the constituency where it gained a residual seat with the \emph{smallest} decimal part; with priority, the seat is reassigned \emph{in that same constituency} to the deficit unit with the largest \emph{unused} decimal part (a decimal part already consumed to gain a residual seat is not reusable, and a transferred seat marks the receiving unit's decimal part in that constituency as used). If no deficit unit has an unused decimal part in any constituency where the surplus unit holds seats, the implementation falls back to the last sentence of \arth{}: subtraction and attribution in \emph{different} constituencies (never needed in 2022 at the constituency level). A guard bounds the loop; in practice it terminates in a handful of moves (Section~\ref{sec:names-method}).

\subsection{Phase 4 --- Art.~83-\emph{bis}}

For each constituency, each list's seats are spread over the multi-member colleges using the college quotient (the truncated ratio of the summed college tallies of the \emph{admitted} lists to the college's predetermined seats, per the official practice documented in the minutes, Section~\ref{sec:validation}): integer parts first, then decimal parts with the static exclusion of the operational practice, then the compensation prescribed by the last sentence of Art.~83-\emph{bis}(1). This compensation differs pointedly from that of \arth{}/\arti{}: the surplus list (the one with the most surplus seats; ties broken by the smaller decimal part) gives up the seat \emph{in the college where it gained it with the smallest decimal part}, while the deficit list (the one with the most missing seats; ties broken by the larger unused decimal part) receives it \emph{in the college where it has its own largest unused decimal part}. These two colleges in general do not coincide, with no same-college priority. Consistently with the law, realised college sizes may therefore deviate from the predetermined ones; this actually happens in 2022 (Section~\ref{sec:results}). As discussed in Section~\ref{sec:validation}, implementing this rule faithfully (rather than reusing the same-constituency-priority rule of \arth{}) is essential to reproducing the official proclamations.

\paragraph{Why Art.~83-\emph{bis} is held fixed at the operational practice.} The A/B/C comparison of this paper varies the interpretation of \arth{} and \arti{} only; Art.~83-\emph{bis} is executed identically (with the static-exclusion reading of the operational practice) in every run. This is a deliberate controlled-variable choice, not a claim that Art.~83-\emph{bis} is unambiguous: its text has the same integer-parts/decimal-parts/exclusion structure and plausibly admits an analogous family of readings. Holding it fixed (a)~isolates the effect of the constituency-level ambiguity, so that every college-level difference reported in Section~\ref{sec:names} is the propagation of an upstream difference through one and the same college procedure, and (b)~pins the college level to the reading documented in the official materials (the college-level worksheets reproduced, for the 2018 election, in the 2020 report of the Committee on Elections~\cite{giunta2020}, and the Electoral Manual's description~\cite{manuale2022}), which is also the reading under which our end-to-end run reproduces the official 2022 proclamations name by name (Section~\ref{sec:validation}). Exploring the ambiguity of Art.~83-\emph{bis} itself would multiply the space of variants combinatorially; we leave it as future work and note only that, as a consequence, the difference counts of Section~\ref{sec:results} are \emph{lower bounds} on the freedom the text as a whole leaves open.

\subsection{Phase 5 --- Arts.~84--85}
\label{sec:impl-f5}

Proclamation proceeds in list order, skipping candidates already elected in a single-member district. Multiple candidacies (a candidate heading several college lists, allowed up to five) are resolved by Art.~85: the candidate takes the seat in the college where their list has the lowest percentage tally. Because resolving one multi-candidacy can promote a lower-ranked candidate who is themselves multi-elected, the implementation iterates the resolution to a fixed point, reached in a few iterations on the 2022 data (the official minutes record 19 multi-elected candidates at the final check). When a list exhausts its candidates in a college, the implementation follows the fallback chain of Art.~84, whose orderings are two-tiered: colleges (and, at para.~4, constituencies) with \emph{unused} decimal parts come first, in decreasing order, followed by those whose decimal parts have already given rise to a seat, again in decreasing order, and each decimal part supports at most one recovered seat, so consecutive missing seats move down the ranking rather than piling up in the same college. Concretely: other colleges of the same constituency (para.~2), the constituency ranking of the list's non-elected district candidates (para.~3), other constituencies by the same two-tier ordering of constituency-level decimal parts, one seat per constituency, each assigned within the designated constituency as in para.~2 (para.~4), and the district rankings of the other constituencies (para.~6). Paragraphs~5 and~7 (fallback to other lists of the same coalition) are not implemented because no 2022 case reaches them. One further point of official semantics, taken from the minutes: the recovered seat is assigned to the first candidate of the target college's list \emph{even if that person is already elected elsewhere}; the resulting multiple election is resolved afterwards, per Art.~85, by proclaiming the person in the college with the lower list percentage and letting the next candidate take over the vacated seat. In the 2022 simulation the chain fires four times, for six seats in total, and each resolution matches the official proclamations and the official minutes decision by decision, including the recovered Campania~1 seat that identifies Alessandra Todde in Lombardia~2 although she is already elected in Sardinia, moves her there (7.6198\% $<$ 21.8028\%), and lets Susanna Cherchi take over the Sardinian seat.

\subsection{What the implementation does \emph{not} decide}

Two aspects are deliberately kept outside the code. First, the choice among A, B and C is an input, not a hard-coded interpretation: the point of the study is to compare them. Second, wherever the law itself is silent (the rare compensation scenarios identified in~\cite{crafa2023}, e.g.\ exact ties of unused decimal parts across constituencies), the implementation resolves the tie by the statutory ``larger national tally, then draw'' rule and logs the event; none of these scenarios occurs on the 2022 data, so no result below depends on such a choice.

\section{Validation against the official outcome}
\label{sec:validation}

We validate at four levels of increasing strictness. The comparison target is the official outcome as published by the Chamber of Deputies: the national apportionment decided by the National Central Electoral Office, and the individual proclamations recorded in the Chamber's linked open data~\cite{daticamera}.

\paragraph{Level 1: exact national apportionment.} Phase~2, executed on the raw municipal data with thresholds computed rather than assumed, reproduces the official apportionment \emph{exactly}, both across units (letter~(f)) and within coalitions (letter~(g)); see Table~\ref{tab:apportionment}. This is a strong end-to-end check: it certifies simultaneously the vote aggregation, the pro-rata reallocation of candidate-only votes, the threshold logic (including the minority rule and the 1\% coalition-tally bar) and the truncated-quotient apportionment.

\begin{table}[t]
\centering\small
\caption{National apportionment of the 245 proportional seats: simulation vs.\ official. The simulation values are \emph{computed} from the raw municipal data; they coincide exactly with the official values.}
\label{tab:apportionment}
\begin{tabular}{lrr@{\hspace{2.5em}}lrr}
\toprule
\multicolumn{3}{c}{Units --- Art.~83(1)(f)} & \multicolumn{3}{c}{Lists --- Art.~83(1)(g)} \\
\cmidrule(r{2em}){1-3}\cmidrule{4-6}
Unit & Tally (votes) & Seats & List & Tally (votes) & Seats \\
\midrule
Centre-right coalition & 12{,}048{,}826 & 114 & Fratelli d'Italia & 7{,}305{,}018 & 69 \\
Centre-left coalition & 7{,}170{,}619 & 68 & Partito Democratico & 5{,}358{,}462 & 57 \\
Movimento 5 Stelle & 4{,}336{,}660 & 41 & Movimento 5 Stelle & 4{,}336{,}660 & 41 \\
Azione--Italia Viva & 2{,}187{,}296 & 21 & Lega & 2{,}464{,}487 & 23 \\
SVP--PATT & 117{,}010 & 1 & Forza Italia & 2{,}279{,}321 & 22 \\
 & & & Azione--Italia Viva & 2{,}187{,}296 & 21 \\
 & & & Alleanza Verdi e Sinistra & 1{,}018{,}344 & 11 \\
 & & & SVP--PATT & 117{,}010 & 1 \\
\midrule
Total & & 245 & & & 245 \\
\bottomrule
\end{tabular}
\end{table}

\paragraph{Level 2: national list tallies.} After the pro-rata reallocation, the national tally of every list agrees with the officially published one within $\pm0.06\%$. The residual discrepancy is a property of the input, not of the algorithm: the municipal open-data file is very close to, but not identical with, the definitive certified counts. This noise level matters for interpreting Level~3.

\paragraph{Level 3: name-by-name comparison.} We match every simulated MP against the official proclamations (mandates starting October 2022) by name tokens, order-insensitively; where one source records middle names that the other omits (three cases, e.g.\ ``Giulio Tremonti'' vs.\ ``Giulio Carlo Danilo Tremonti''), the match is completed by requiring the same college and list. Two typographic inconsistencies in the ministry's candidate file (the same candidate recorded with two different birth dates in different colleges, A.~Soumahoro, C.~Taibi) are reconciled to the majority date at load time; without this repair the identity key splits, a multiple candidacy goes undetected by Art.~85, and one candidate is elected twice. The result, for the reference variant C, is summarised in Table~\ref{tab:names}: 389 of 391 seats (99.5\%) are matched by name, including \emph{all} 146 single-member winners (which depend only on district pluralities) and 243 of 245 proportional seats (99.2\%). No candidate is elected twice, all 27 constituencies have exactly their statutory seat totals, and all corner-case chains (the multi-candidacy resolutions and the six exhausted-list seats) reproduce the officially proclaimed names.

\begin{table}[t]
\centering\small
\caption{Name-level validation of variant C against the official proclamations.}
\label{tab:names}
\begin{tabular}{lrrr}
\toprule
& Seats modelled & Matched by name & \\
\midrule
Single-member districts & 146 & 146 & (100\%) \\
Proportional seats & 245 & 243 & (99.2\%) \\
\midrule
Total & 391 & 389 & (99.5\%) \\
\bottomrule
\end{tabular}
\end{table}

The two unmatched seats are two of the four seats of the 2022 allocation that the Chamber's Committee on Elections (\emph{Giunta delle elezioni}), the body constitutionally charged with validating the election, placed under formal investigation on 22~July 2025, in both cases on the machinery of Arts.~83-\emph{bis}--85 that assigns seats and persons to colleges~\cite{giunta2025}.

The first concerns one FdI seat in Toscana: our run elects Irene Gori in college P01, whereas the official proclamation elects Francesco Michelotti in P02 (simulation and proclamation agree on the number of FdI seats in every Tuscan college). The minutes of the Regional Central Electoral Office (\emph{Ufficio centrale circoscrizionale}, UCC) for Toscana report the FdI lists' percentage tallies (the very quantities that Art.~85 uses to resolve multiple candidacies) \emph{twice, with different values}. The tallies computed from the figures set out in the body of the minutes (p.~24, \S11 of that \emph{verbale}) place Alessandro Amorese and Irene Gori in P01 and push Chiara La~Porta down to P02 (where her percentage is the lower of the two, so Art.~85(1) elects her there and P01 has a free seat for Gori); a second set of tallies reported later in the same document (p.~51) gives La~Porta 28.20\% in P01 and elects her there, so that P02 (after the list's next two names, Fabrizio Rossi and Chiara Colosimo, are skipped as already elected in single-member districts) falls to Michelotti. Our implementation computes the Art.~85 percentages from the underlying tallies and so reproduces the first reading; the official proclamation followed the second. The discrepancy did not go unnoticed at the time: on Irene Gori's appeal, the Committee resolved to open an \emph{istruttoria} into the proclamation of Michelotti in Toscana~P02 and to defer the final determination of La~Porta's college of election, expressly on the ground of ``the modalities of application of the rules on multiple election (the combined provisions of Art.~85, paragraphs~1 and 1-\emph{bis})'' and of ``the uncertainties connected to the method of calculation followed by the Toscana UCC for the different determination of the percentage figures reported in two points of the [minutes]''~\cite{giunta2025}.

The second concerns the last seat of Alleanza Verdi e Sinistra, whose lists are headed in four colleges by Aboubakar Soumahoro. Our run places AVS's Puglia seat in college P02 (through the Art.~83-\emph{bis} compensation), proclaims Soumahoro there (his lowest list percentage), and the seat thereby freed in Lombardia~1--P01 elects the next name, Giovanni Paglia. The official proclamation placed Soumahoro in Emilia-Romagna~P02, Eleonora Evi in Lombardia~1--P01 and Elisabetta Piccolotti in Puglia~P04. The choice between the two Puglia colleges rests on the comparison of AVS's unused decimal parts in P02 and P04, a comparison that, in the Committee's own account, already at proclamation rested on ``a very small margin'', and that the Committee's verification, having found inversions of list votes in the returned counts, could not settle without further investigation~\cite{giunta2025}. Which Puglia college (P02 or P04) receives the AVS seat, and in which colleges Soumahoro and Evi are proclaimed, is precisely the object of the second \emph{istruttoria} of the same session, on Piccolotti's proclamation with deferred determination of Evi's and Soumahoro's colleges of election~\cite{giunta2025}. The same session opened two further \emph{istruttorie} on neighbouring questions (FdI seats between Emilia-Romagna colleges under the combined Art.~84(2)/85 rules, and a tally-correction case in Puglia); the five cases in which our simulation elects the same person in a different college (La~Porta, Soumahoro, Evi, Vinci, Lucaselli) all fall within the seats they cover.

Neither residual is an instance of the paper's central ambiguity: both live in the Arts.~83-\emph{bis}--85 machinery (college attribution and resolution of multiple candidacies), not in the order-dependent exclusion clause of \arth{}, and both are triggered by unsettled points of the official computation (internally inconsistent percentage tallies in one case, a college attribution resting on a very small decimal margin in the other) that the Chamber's own verification body is still adjudicating three years on. We accordingly classify them not as realisations of the \arth{} ambiguity but as further entries in the inventory of text--practice gaps (Section~\ref{sec:discussion}). We keep our upfront rule for skipping district winners (it is equivalent to the official identify-then-substitute procedure of Art.~86(1) everywhere in 2022 except in the contested Toscana college) and report the divergences exactly as what they are: not defects to be tuned away, but the empirical footprint of points where the official computation itself is genuinely unsettled.

Separately, we record one further difference between our run and today's open data that is not an ambiguity at all: for the Calabria P01 college, the simulation reproduces the \emph{2022 proclamation} (Orrico elected in the Cosenza district, Scutell\`a taking the P01 list seat via Art.~84(3)), whereas the current open data reflect the seat reassignment ordered after the 2025 judicial recount of the Cosenza district. This is a difference between the 2022 and the present composition of the Chamber, not a simulation error, and it does not enter the 389/391 count (which is taken against the 2022 proclamations).

\paragraph{Level 4: step-level comparison with the official minutes.} The minutes of the National Central Electoral Office (\emph{verbale UECN}, Modello n.~89 E.P.) are published as scanned images; we transcribed the three available parts (the apportionment worksheets and the section on substitutions, exhausted lists and multiple elections) by OCR, proof-reading the extracted values used in the comparison against the page images, and compared them with the simulation \emph{step by step}. The transcriptions are archived with the deposit (Section~\ref{sec:repro}), so every extracted value can be re-checked against the published scans; since the comparison is value-by-value, a residual transcription error would surface as an isolated mismatch inexplicable by data noise, and none was observed. The agreement is essentially exact. The constituency-level decimal parts tabulated in the minutes match ours within $\pm 0.001$ (e.g., for M5S: Sardegna 0.678213 official vs.\ 0.6781 ours; Veneto 2 0.746397 vs.\ 0.7466; Lombardia 2 0.635208 vs.\ 0.6348), and their used/unused flags match exactly; the college-level percentage tallies used by Art.~85 agree to within the residual data noise (typically $10^{-3}$, occasionally $10^{-2}$; e.g.\ Carfagna, Puglia P04: 4.625339 official vs.\ 4.6246 ours; Todde, Lombardia 2 P01: 7.629735 vs.\ 7.6198), and every Art.~85 decision recorded in the minutes coincides with the simulation's; the college decimal parts match once the college quotient is computed on the admitted lists (M5S, Lombardia 2 P01: 0.340704 official vs.\ 0.3398 ours; the all-lists reading would give 0.3048); and the recovery chains coincide decision by decision: the minutes assign the stranded Campania 1 seat to Lombardia 2 ``with the largest unused decimal part, 0.635208'', identify Alessandra Todde there, resolve her multiple election in favour of Lombardia 2, and let Cherchi take over in Sardinia, exactly as the simulation does; the Veneto 1 FdI seat goes to Lazio 2 -- P02 and Paolo Pulciani in both. The residual noise on decimal parts and percentages is the footprint of the near-definitive municipal file and never crosses a decision margin in the minutes we transcribed.

\paragraph{The diagnostic power of name-level validation.} The name-by-name comparison is not only a validation score: it is strict enough that a single misread sentence of the statute produces a systematic disagreement pattern that input noise cannot explain. Three points of the procedure were settled in exactly this way, against the statute and the minutes: (1) the last sentence of Art.~83-\emph{bis}(1) selects donor and recipient colleges \emph{independently}: importing the same-constituency priority of \arth{} at this level displaces two dozen college placements; (2) the orderings of Art.~84 are two-tiered (unused decimal parts first) with one recovered seat per decimal part, and the recovered candidate may be already elected elsewhere, the conflict being resolved afterwards under Art.~85; (3) the college quotient is computed on the admitted lists only, as the decimal parts in the minutes prove. The methodological point is that reproducing an electoral law at the level of \emph{names} is discriminating enough to detect each of these misreadings, and that the official minutes, once machine-readable, adjudicate every remaining doubt. Crucially, all three points concern the \emph{fixed machinery} (Art.~83-\emph{bis}, Art.~84 and the college quotient) that every A/B/C run executes identically (Section~\ref{sec:impl-h}); they are a controlled variable common to all variants, and they were settled against an \emph{independent} record (the minutes), not against the proclamation names. Settling them therefore cannot manufacture the A-vs-C differences that are this paper's object: those arise entirely from the interpretation of \arth{}/\arti{}, which no validation step ever tuned. The fixed procedural machinery was settled first and then held frozen while the interpretive variants were compared.

\paragraph{Validation protocol.} Because the previous paragraph reports corrections made in the course of validation, we state precisely what was fixed when, so that the reader can judge, and exclude, the risk of tuning the implementation to the outcome. (i)~The data reconstruction and Phases~1--2 (vote aggregation, pro-rata reallocation, thresholds, national apportionment) were implemented from the statute and the data documentation alone, and never adjusted afterwards; their validation targets (Levels~1--2) are aggregate official figures, not names. (ii)~Exactly three implementation decisions were corrected after the name-level comparison first ran: the two Art.~83-\emph{bis}/Art.~84 rules and the college-quotient basis listed above. Each was settled against the official minutes (an independent record of the intermediate computation) and not by searching for name agreement; each corresponds to an identifiable sentence of the statute or of the minutes, quoted where introduced. (iii)~All three belong to the fixed machinery that every variant executes identically; none touches the implementation of \arth{}/\arti{}, which was written once from the statutory text and never modified. (iv)~The variant experiments (A~vs.~C, the order studies, B) were run after the fixed machinery was frozen, and their results were never used to revise it. (v)~The archived deposit preserves the audit logs of every run and the scripts of every experiment (Section~\ref{sec:repro}), so the entire sequence can be re-executed end to end.

We regard Levels 1--4 together as strong evidence that the implementation is a faithful executable model of the law as officially applied: it is exact wherever the input data are exact, and deviates only where the input deviates, in the direction and magnitude that the deviation predicts.

\section{Results: the ambiguity on real data}
\label{sec:results}

All experiments run the full pipeline end-to-end (from raw votes to named MPs), changing only the interpretation of \arth{}/\arti{} (A, B or C) and, where indicated, the processing order of the constituencies. The reference order is the order in which the official UECN worksheets tabulate the 27 constituencies (the order of Appendix~\ref{app:letth}); it differs from the order of Table~A annexed to the \dpr{} in one point: the worksheets place Trentino-Alto Adige last, after Sardegna, whereas Table~A lists it seventh. We adopt the worksheets' order as the reference (hereafter also ``the official order'') because it is the one documented in the official tallying record; under interpretation C the choice is immaterial (Proposition~\ref{prop:C}), and the order study below quantifies its effect under A. The permuted order is the reversal of the reference order. Results are compared as sets of elected persons and, for persons elected in both runs, as their college of proclamation.

\begin{table}[t]
\centering\small
\caption{Effect of the interpretation of \arth{} on the 2022 election. ``Replaced'' counts deputies elected under one run and not under the other (per side); ``displaced'' counts deputies elected in both runs but proclaimed in different colleges. The order experiment compares the official and the reversed constituency order; ``vs.\ C'' compares each variant (official order) with the reference variant C.}
\label{tab:variants}
\begin{tabular}{lcccc}
\toprule
& \multicolumn{2}{c}{Order effect} & \multicolumn{2}{c}{vs.\ C (official order)} \\
\cmidrule(lr){2-3}\cmidrule(l){4-5}
Interpretation & replaced & displaced & replaced & displaced \\
\midrule
A (sequential) & $6+6$ & 1 & $4+4$ & 0 \\
B (Mattarella-style) & \multicolumn{2}{c}{not measured} & $11+12$ & 1 \\
C (operational practice) & 0 & 0 & --- & --- \\
\bottomrule
\end{tabular}
\end{table}

Table~\ref{tab:variants} summarises the headline findings; Section~\ref{sec:names} identifies every affected deputy by name and re-derives each change step by step.

\paragraph{Algorithm A is order-dependent at a politically visible scale.} Reversing the constituency order changes the territorial placement of enough residual seats that \textbf{12 people enter or leave the Chamber} (6 elected only under the official order, 6 only under the reversed order) and one more deputy is proclaimed in a different college. Both runs execute the same admissible reading of the text: the law neither fixes an order nor makes the outcome order-invariant under interpretation A. To quantify the phenomenon beyond a single permutation, we additionally ran Algorithm A under \textbf{200 uniformly random constituency orders} (same votes, same algorithm). Relative to the practice outcome (C), the number of deputies elected under the sampled order but not in practice ranges from 0 to 13 (mean 5.4, median 6), the number of practice deputies who lose their seat from 0 to 14 (mean 5.8; the two counts differ whenever the order strands seats, a phenomenon quantified below), and the number of displaced deputies from 0 to 4 (mean 0.8); relative to A under the official order, the entering side ranges from 3 to 13 (median 7). The 200 orders produce \textbf{161 distinct outcomes}; at least one random order reproduces the practice parliament exactly, and none reproduces the official-order run of A. The reversal experiment of Table~\ref{tab:variants}, and the names of Section~\ref{sec:names-Arev}, are therefore a typical instance of the phenomenon, not an extremal one. This is the empirical confirmation, at full scale, of the core finding of~\cite{crafa2023}.

\paragraph{Sampling design of the order experiment.} The 200 orders are drawn uniformly at random (Fisher--Yates shuffles from fixed, published seeds; the sampling script is archived with the deposit, Section~\ref{sec:repro}), so the counts above are descriptive statistics of the uniform distribution over the $27!$ possible orders, estimated from 200 samples. To gauge the sampling error we extended the experiment, under the same seed family, to 1{,}000 orders (the first 200 coincide with those reported above): the mean number of entering deputies relative to practice is 5.20 (bootstrap 95\% interval $[5.03, 5.36]$), the median 5 ($[5,6]$), the range still 0--13 (0--14 on the leaving side, mean 5.53), and the 1{,}000 orders produce 560 distinct outcomes; 74 of the 1{,}000 orders (7.4\%) reproduce the practice parliament exactly, and none reproduces the official-order run of A. Two caveats delimit what these numbers mean. First, they estimate the diversity of outcomes under uniformly random orders, not real-world probabilities: the law fixes neither a processing order nor a distribution over orders, and a returning office would presumably follow some conventional order rather than a random one. The uniform distribution is a device for exploring the space of admissible executions, and we accordingly report ranges and counts rather than significance tests. Second, the number of distinct outcomes is a lower bound that grows with the number of samples (161 at 200 orders, 560 at 1{,}000): it measures the sampled diversity of the outcome space, not the total number of outcomes the text admits.

\paragraph{Under arbitrary orders, A strands seats --- and can move one between parties.} The sweep exposes a defect of the sequential reading that no documented execution shows. In \textbf{292 of the 1{,}000 sampled orders (29.2\%)} Algorithm~A cannot assign one residual seat (249 orders) or two (43 orders): the Chamber would convene with 389--390 of the 391 modelled members. The mechanism is fully statutory, and the first sampled order shows it at work. Marche is processed last, with 3 residual seats; by then M5S stands at 42 of its 41 national seats (integer parts are never capped, since the exclusion clause of \arth{} governs only the residual distribution, and Toscana's integer parts pushed it past its entitlement), and A--IV is saturated at 21 of 21. The two eligible units, CD and CS, take one residual seat each, and the third finds no taker: each unit has a single decimal part per constituency, so a unit still below its national total (CD, at 112 of 114) cannot take a second residual seat there. The compensation of \arth{}, which can only transfer surpluses, recovers one seat (M5S to CD) but cannot create the missing one. This is the same class of textual silence documented for Algorithm~B in Section~\ref{sec:names-B}, and the empirical realisation, at scale, of the underspecified compensation scenarios identified analytically by Crafa~\cite{crafa2024}: absent under both documented orders of A, it occurs in 29\% of uniformly random admissible executions, always at the expense of the centre-right lists (Lega and/or Forza Italia in every incomplete order). In \textbf{89 further orders (8.9\%)} the pipeline completes (all 391 seats filled) but with a \emph{different party composition}: the seat lost by a centre-right list at the constituency level is refilled at the college level by the Art.~83-\emph{bis} machinery on a different list (Forza Italia $-1$/Partito Democratico $+1$ in 46 orders, Forza Italia $-1$/M5S $+1$ in 35, Lega $-1$/M5S $+1$ in 6, Lega $-1$/Partito Democratico $+1$ in 2), through the same mechanism as the recreated PD seat documented for B in Section~\ref{sec:names-B}. Only \textbf{619 of the 1{,}000 sampled orders} reproduce the official party totals exactly. As with the B name lists, these mechanical completions rest on choices the text does not define; the counts of stranded and migrated seats, however, are exact consequences of interpretation A and of the documented compensation machinery (the per-order results are archived with the deposit, Section~\ref{sec:repro}).

\paragraph{A per-deputy view of the order experiment.} The same 1{,}000 runs, read per person rather than per order, show who the order sensitivity actually puts in play. Across the sampled orders, 458 distinct persons are elected (to 391 seats when the order completes the Chamber, 389--390 when it strands seats): 339 deputies of the practice parliament (including all 146 single-member winners, whose election does not depend on \arth{}) are elected under \emph{every} sampled order, while 119 persons are elected under some orders and not others; Appendix~\ref{app:perm} lists them all with their frequencies. Of the practice parliament, 52 deputies miss election under some orders, though none under all (every member of the practice parliament is elected under at least 300 of the 1{,}000 orders), and conversely 67 persons who hold no seat in practice enter the Chamber under some order. The frequency distribution is markedly bimodal: 42 of the 119 are elected under more than 900 of the 1{,}000 orders and 59 under fewer than 100, so the typical effect of changing the order is to trade a specific near-safe deputy for a specific long-shot substitute; but for a handful of persons the seat approaches a coin flip: Marco Squarta (FdI, Umbria) is elected under 399 of the 1{,}000 orders, Umberto Bossi (Lega, Lombardia 2) under 386, Rosa D'Amelio (PD, Campania 2) under 382, Devis Dori (AVS, Lombardia 2) under 307, Valeria Alessandrini (Lega, Umbria) under 301, Catia Polidori (FI, Umbria) under 300. One case stands out: Roberto Rampi (PD, Lombardia 2) holds no seat in practice, yet is elected under 707 of the 1{,}000 orders: under a uniformly random admissible execution of interpretation A he is a deputy more often than not. As with the parliament counts above, these frequencies describe the uniform sample of admissible orders, not real-world probabilities.

\paragraph{A near-exhaustive extension to 15 million orders.} The 1{,}000-order sweep is a sample, and the number of order-contingent persons it identifies is a lower bound that could still grow. To test how far, we exploited the structure of the pipeline: the processing order enters only through Phase~3 (the territorial distribution of \arth{}/\arti{}), while Phases~4--5 (college placement and proclamation, Arts.~83-\emph{bis}, 84--85) are a deterministic function of its output. Phase~3 costs $\sim$0.8\,ms, some two orders of magnitude less than a full run, so we evaluated \textbf{15 million uniformly random orders}, deduplicated them on the order-dependent Phase~3 state (obtaining 41{,}953 distinct configurations), and ran Phases~4--5 once per distinct configuration; a consistency test confirms that this factored evaluation reproduces the full pipeline exactly on every order checked. The set of order-contingent persons rises from the 119 of the 1{,}000-order sample to \textbf{140} and then closes: across the 41{,}953 configurations the last new order-contingent deputy appears after about 17{,}500 of them, and none of the remaining $\sim$24{,}500 adds another; the number of distinct Phase-3 configurations is itself approaching saturation (fewer than $0.05\%$ of the final million orders produce a configuration not already seen). Across all 15 million orders, 469 distinct persons are elected, 329 under \emph{every} order (including the 146 single-member winners), and 140 are order-contingent: 62 deputies of the practice parliament who lose their seat under some order, and 78 who hold no seat in practice but enter under some order. Referred to the pool that actually competes for the 245 proportional seats --- the \textbf{1{,}139 candidates of the eight lists admitted to the distribution} --- the arbitrariness of the processing order thus leaves \textbf{764 of them (67.1\%) elected under no order at all}, seats a stable core, and decides the fate of the remaining $\sim$12\%. This is a near-exhaustive exploration, not a proof: an exact enumeration of the outcome set is obstructed by the path-dependence of the statutory compensation (the residual-and-compensation state does not collapse over the $27!$ orders), so the closure of the order-contingent set is strong empirical evidence, not a theorem. The scripts and the per-configuration outputs are archived with the deposit (Section~\ref{sec:repro}).

\paragraph{Algorithm C is order-independent.} Under interpretation C, the official order, the reversed order and 50 additional uniformly random orders all produce \emph{identical} parliaments: zero replaced, zero displaced. This is not an accident of the data:

\begin{proposition}[order invariance of C in the absence of ties]
\label{prop:C}
Fix the votes and the seat entitlements. Assume that (i) in every constituency the decimal parts of the quotients of attribution of the admitted units are pairwise distinct, and (ii) at every step of the compensation procedure the statutory selection criteria (number of surplus seats, decimal parts, national tallies) identify a unique choice. Then the outcome of Algorithm~C, followed by the statutory compensation, is the same for every processing order of the constituencies.
\end{proposition}

\begin{proof}[Proof sketch]
Integer parts are computed constituency by constituency from order-free inputs (tallies and per-constituency seat counts), so their totals are order-free; the static exclusion set $E$ is a function of those totals, hence order-free; residual seats are then assigned \emph{within} each constituency by a rule (largest decimal parts among units not in $E$, ties broken by national tally) that depends only on that constituency's own quantities and on $E$ (there is no cross-constituency state), so the union of the residual assignments is order-free; finally, the compensation is driven by comparisons of order-free quantities (surplus counts, decimal parts, national tallies), which under (ii) determine a unique sequence of transfers. Hypotheses (i)--(ii) hold on the 2022 data: we verified computationally that no tie occurs in any comparison actually performed. Were they to fail, the law's own tie-breaking (larger national tally, then draw) applies, and order invariance would hold up to the outcome of the draws.
\end{proof}

\paragraph{Ties.} One point of the tie-handling deserves the main text, because it matters for any legislative formulation of the procedure. For strict order-invariance to survive ties, the draw must not consume a sequentially shared pseudo-random state, which would make later draws depend on the processing order; a canonical tie-break (e.g.\ seeding the draw on the invariant identity of the list or constituency rather than on its position in the sweep) restores exact order-invariance. No draw is triggered on the 2022 data, so this affects no result here, but a statutory algorithmic formulation of the procedure should specify it.

\paragraph{A and C elect different parliaments.} Under the same (official) order, interpretations A and C differ in \textbf{8 deputies} (4 elected only under A, 4 only under C). Since C is the interpretation applied in practice while A is the most immediate sequential reading of the text (in the algorithmic sense of Section~\ref{sec:legal-scope}), these seats quantify, on real data, the gap between the enacted text and the executed procedure.

\paragraph{Algorithm B loses seats.} Interpretation B leaves \textbf{2 seats of the centre-right coalition unassigned}: when the residual-seat sweep reaches Sicilia~2 and Sardegna, every admitted unit is either saturated or has already used its decimal part there, and the seats cannot be attributed by any rule stated in the text (Section~\ref{sec:names-B}). Because B never over-assigns, there are no surplus units, and the statutory compensation (which can only \emph{transfer} seats from surplus to deficit units) has nothing to transfer. We phrase the consequence carefully. Within our mechanical implementation of B, which executes exactly the rules the text states, the two seats remain unassigned and the Chamber would convene with 398 of 400 members. Under the Mattarella law, whose Algorithm-B-like procedure produced the same situation in 1994, 1996 and 2001, the stranded seats were recovered in practice by assigning them in the constituencies where the entitled lists had their highest unused remainders~\cite[p.~66]{camera2019}; the text in force in 2022 contains no such suppletive rule for this step, and whether the National Electoral Office could legitimately construct one by analogy is a question of law outside our scope. What the experiment establishes is that interpretation B \emph{requires} such an extra-textual rule on real data: the defect is not an artefact of the older electoral maps. Interpretation A, as the order sweep above shows, incurs the same textual silence under 29\% of the sampled orders. B is of the same kind and worse in degree: run end-to-end under the same 1{,}000 sampled orders, it fails to fill one to three seats in 45.7\% of them and reproduces the official party totals in only 15.7\%, against A's 29.2\% and 61.9\% (the Art.~83-\emph{bis} college machinery, as in the recreated PD seat of Section~\ref{sec:names-B}, often refills a constituency-level deficit on a different list; the per-order results are archived with the deposit, Section~\ref{sec:repro}).

Three questions about interpretation B must be kept separate, and the paper treats them at three different levels of confidence: (i)~\emph{what the text literally executes}: the mechanical run just described, which strands two seats; this is established computationally, and is as robust as the A-vs-C results. (ii)~\emph{What a completed B would produce}: the next paragraph describes the completion suggested by the Mattarella-era practice, and Table~\ref{tab:names-B} the resulting names; both are conditional on an extra-textual completion rule that the text does not define, and are reported as indicative only. (iii)~\emph{Whether B is legally admissible}: a doctrinal question (Section~\ref{sec:legal-scope}), on which the stranded seats are themselves perhaps the strongest evidence against, but which our experiments cannot settle.

\paragraph{What a historically completed B would do.} For comparability with the completing variants, we describe (without adopting it) the completion the Mattarella-era practice would apply. The two stranded seats both belong to the centre-right, the sole deficit unit; the historical rule recovered such seats by assigning them, one each, in the constituencies where the entitled unit holds its largest still-\emph{unused} remainders, which would restore the centre-right to 114 and the Chamber to 400 members. This patch is itself extra-textual and not uniquely fixed: it must decide which remainders count as unused, and how the recovery propagates to the within-coalition level (\arti{}, where B also leaves Forza Italia and the Lega one seat short) and to the colleges (Art.~83-\emph{bis}). A ``completed B'' is therefore one reconstruction among several rather than a canonical variant, which is why we report B as the text literally executes it (398 seats) and treat the named differences of Table~\ref{tab:names-B} as indicative rather than exact.

\paragraph{Party totals are invariant under the documented executions; people and places are not.} In every documented execution (C under any order, A under the reference and the reversed order), the statutory compensation procedure restores the national totals of every coalition \emph{and} of every list to the values of Table~\ref{tab:apportionment}. Under the executions documented in practice, therefore, the ambiguity of \arth{} has \emph{no effect on the party composition} of the Chamber; its entire effect is on \emph{which} candidates of a given list are elected, and \emph{in which territory}. The order sweep above marks the boundary of this invariance: under arbitrary orders of the sequential reading, seats can go unassigned or migrate between lists. This is exactly the reading of the problem proposed by Crafa: the indeterminacy is constitutionally serious not because it could flip a majority, but because the identity of the elected representatives (and the territorial link between voters and their representatives) should not depend on an arbitrary implementation choice left open by the text. The point has a precise constitutional anchor: under Article~67 of the Italian Constitution every member of Parliament ``represents the Nation'' and serves ``without a binding mandate'': the unit of representation is the person, not the party share. The free mandate also makes the invariance of party totals an election-day property in a second, practical sense: each deputy remains free to change parliamentary group, and does so in practice (in the very legislature elected by the votes studied here, 39 deputies had changed group by mid-term, with further changes officially recorded through 2026, and preceding legislatures saw hundreds of such moves~\cite{openpolis2025,camera1083}), so that \emph{which} person fills a list's seat can, over the course of the legislature, move votes between groups in a way the seat-allocation arithmetic never registers.

\section{From seats to names: the elected deputies under each interpretation}
\label{sec:names}

This section identifies, by name, every deputy whose election depends on the interpretation of \arth{}, and re-derives each change step by step. All arithmetic can be checked against the appendices: Appendix~\ref{app:letth} publishes the complete \arth{} computation (constituency quotients, integer and decimal parts for all 27 constituencies and all 5 admitted units), Appendix~\ref{app:assh} the final unit-level assignments of all four runs, and Appendix~\ref{app:i} the inputs and final assignments of the within-coalition distribution (\arti{}). Nothing beyond these tables, the college tables printed in the case studies, and the candidate lists quoted in the text is needed to re-verify every name.

\subsection{Method: how a unit-level seat becomes a named deputy}
\label{sec:names-method}

A difference between two interpretations always originates at \arth{} (a coalition or single list gets a seat in a different constituency) and propagates through three deterministic steps:
\begin{enumerate}
\item \emph{\arti{}}: within the affected coalition, the changed constituency seat count changes the constituency quotient of the coalition's lists (the quotient is the truncated ratio of the sum of the list tallies to the coalition's seats there), hence possibly the list that takes the seat; the within-coalition compensation can move further seats between constituencies.
\item \emph{Art.~83-\emph{bis}}: within the constituency, the changed per-list seat count changes which college receives the seat (by the college decimal parts, with college-level compensation).
\item \emph{Arts.~84--85}: within the college, the seat goes to the highest-ranked candidate of the closed list who is not elected in a single-member district and not proclaimed elsewhere; multiple candidacies are resolved to the college with the lowest list percentage (Art.~85), so a seat moving between colleges can drag a multi-candidate with it and free a seat for the next name on another list.
\end{enumerate}

For the sweep of residual seats we describe each run by its \emph{divergence points}: the constituencies where the dynamic exclusion of step (6) causes an assignment different from the pure largest-decimal choice. Everything before the first divergence point is identical in all runs and can be read off Appendix~\ref{app:letth}.

Two facts, valid in all runs, anchor the analysis. First, the integer parts of \arth{} assign 190 of the 245 seats (centre-right (CR) 100, centre-left (CL) 54, M5S 30, Azione--Italia Viva (A--IV) 6, SVP 0), leaving 55 residual seats (both facts follow by summation from Appendix~\ref{app:letth}). Since no unit reaches its national entitlement (114/68/41/21/1) on integer parts alone, the static exclusion set of Algorithm~C is \emph{empty}: under C, all 55 residual seats go to the largest decimal parts unconditionally, and the compensation step then repairs the totals. Second, the compensation moves of all runs are few and are listed exhaustively in Table~\ref{tab:comp}.

\begin{table}[t]
\centering\small
\caption{All compensation transfers at the \arth{} and \arti{} levels, in execution order (surplus unit $\rightarrow$ receiving unit; the value in parentheses is the surplus unit's decimal part in that constituency, which the law requires to be minimal among its residual-seat gains). Under B no transfer is possible: there are no surplus units, and the deficits (CR $-2$ at letter (h); FI $-1$, Lega $-1$ at letter (i)) remain.}
\label{tab:comp}
\begin{tabular}{llll}
\toprule
Run & Level & Constituency & Transfer \\
\midrule
C & (h) & Lombardia 1 & M5S $\rightarrow$ CR \quad (0.5189) \\
C & (h) & Lombardia 3 & M5S $\rightarrow$ CR \quad (0.5532) \\
C & (h) & Lombardia 2 & M5S $\rightarrow$ CR \quad (0.6348) \\
C & (i, CR) & Umbria & FdI $\rightarrow$ FI \quad (0.3577) \\
C & (i, CL) & Veneto 2 & PD $\rightarrow$ AVS \quad (0.5021) \\
C & (i, CL) & Sardegna & PD $\rightarrow$ AVS \quad (0.5753) \\
C & (i, CL) & Campania 2 & PD $\rightarrow$ AVS \quad (0.6230) \\
C & (i, CL) & Lombardia 2 & PD $\rightarrow$ AVS \quad (0.6472) \\
\midrule
A & (h) & Lombardia 1 & M5S $\rightarrow$ CR \quad (0.5189) \\
A & (h) & Lombardia 3 & M5S $\rightarrow$ CR \quad (0.5532) \\
A & (i, CR) & Umbria & FdI $\rightarrow$ FI \quad (0.3577) \\
A & (i, CL) & Veneto 2 & PD $\rightarrow$ AVS \quad (0.5021) \\
A & (i, CL) & Campania 2 & PD $\rightarrow$ AVS \quad (0.6230) \\
\midrule
A$_{\mathrm{rev}}$ & (h) & Lombardia 1 & CL $\rightarrow$ CR \quad (0.2467) \\
A$_{\mathrm{rev}}$ & (h) & Lombardia 3 & M5S $\rightarrow$ CR \quad (0.5532) \\
A$_{\mathrm{rev}}$ & (i, CR) & Umbria & FdI $\rightarrow$ Lega \quad (0.3577) \\
A$_{\mathrm{rev}}$ & (i, CR) & Basilicata & FdI $\rightarrow$ FI \quad (0.4976) \\
A$_{\mathrm{rev}}$ & (i, CL) & Veneto 2 & PD $\rightarrow$ AVS \quad (0.5021) \\
A$_{\mathrm{rev}}$ & (i, CL) & Sardegna & PD $\rightarrow$ AVS \quad (0.5753) \\
A$_{\mathrm{rev}}$ & (i, CL) & Campania 2 & PD $\rightarrow$ AVS \quad (0.6230) \\
\bottomrule
\end{tabular}
\end{table}

\subsection{Interpretation A versus practice (C), official order}
\label{sec:names-AC}

\begin{table}[t]
\centering\small
\caption{Named differences between interpretation A and interpretation C (both under the official constituency order). Every other one of the 391 seats is held by the same person in the same college under both interpretations.}
\label{tab:names-AC}
\footnotesize
\begin{tabular}{llll}
\toprule
& Deputy & List & College \\
\midrule
\multirow{4}{*}{\shortstack[l]{Elected under C,\\ \emph{not} under A}}
& Umberto Bossi & Lega & Lombardia 2 -- P01 \\
& Devis Dori & AVS & Lombardia 2 -- P02 \\
& Sara Ferrari & PD & Trentino-A.A. -- P01 \\
& Maria Stefania Marino & PD & Sicilia 2 -- P01 \\
\midrule
\multirow{4}{*}{\shortstack[l]{Elected under A,\\ \emph{not} under C}}
& Roberto Rampi & PD & Lombardia 2 -- P01 \\
& Diego Binelli & Lega & Trentino-A.A. -- P01 \\
& Romina Mura & PD & Sardegna -- P01 \\
& Pierpaolo Montalto & AVS & Sicilia 2 -- P02 \\
\bottomrule
\end{tabular}
\end{table}

Table~\ref{tab:names-AC} lists the named differences. We now derive them.

\paragraph{The divergence at \arth{}.} Under the official order, interpretations A and C make identical residual-seat choices in the first 25 constituencies: no running total crosses a national entitlement early enough for the dynamic exclusion of A to bite there. The divergence is confined to the tail of the sweep. Following the running totals of Algorithm~A (all decimal parts from Appendix~\ref{app:letth}):

\begin{itemize}
\item By the end of Calabria (23rd), M5S has 36 seats; the integer parts of Sicilia 1, Sicilia 2 and Sardegna (3, 3 and 1) bring it to 43, \emph{above} its national entitlement of 41. (Under C this overshoot is tolerated and repaired at the end; under A it triggers exclusions.)
\item \textbf{Sardegna} (26th; 7 proportional seats, unit-level quotient $q_c = \lfloor 623{,}528/7 \rfloor = 89{,}075$, 2 residual seats). The first residual seat goes to CR (decimal 0.9611, not saturated: 110/114). The second, under C, goes to M5S (0.6781). Under A, M5S is saturated (43/41) and skipped; A--IV is saturated (21/21, decimal 0.3547) and skipped; the seat falls to \textbf{CL with decimal part 0.0062}: the smallest decimal part of the entire election defeats 0.6781.
\item \textbf{Trentino-Alto Adige} (27th; 3 proportional seats, 2 residual). Under C the residuals go to CL (0.8562) and SVP (0.7649). Under A, CL is now saturated (68/68, thanks to the Sardinian seat) and skipped; SVP takes its seat (0.7649, reaching its entitlement of 1); the last seat falls to \textbf{CR with decimal part 0.0124}, defeating CL's 0.8562.
\end{itemize}

After the sweep, A has M5S at 43 ($+2$) and CR at 112 ($-2$): the compensation returns M5S's two smallest-decimal residual gains, Lombardia 1 (0.5189) and Lombardia 3 (0.5532), to CR. Under C, M5S ends at 44 ($+3$) and also returns \textbf{Lombardia 2} (0.6348). The net unit-level difference between A and C is therefore exactly three seats (Appendix~\ref{app:assh}): in \textbf{Lombardia 2}, CR has 6 seats under C but 5 under A, M5S 0 vs 1; in \textbf{Sardegna}, CL has 2 vs 3, M5S 2 vs 1; in \textbf{Trentino-Alto Adige}, CR has 1 vs 2, CL 1 vs 0.

\paragraph{Case study 1: Trentino-Alto Adige --- Sara Ferrari (PD) out, Diego Binelli (Lega) in.}
Trentino-Alto Adige has a single multi-member college (P01, 3 seats), so \arti{} decides everything. \emph{Under C}, CL has 1 seat; the CL quotient is $\lfloor(86{,}459 + 29{,}599)/1\rfloor = 116{,}058$; the quotients of attribution are PD $86{,}459/116{,}058 = 0.7450$ and AVS $0.2550$; the seat goes to PD, whose college list opens with \emph{Sara Ferrari} (elected). CR's single seat: quotient $\lfloor(94{,}823+43{,}109+16{,}945)/1\rfloor = 154{,}877$; decimals FdI 0.6122, Lega 0.2783, FI 0.1094; the seat goes to FdI (and FdI's national surplus of one seat is later returned in Umbria, not here; Table~\ref{tab:comp}). \emph{Under A}, CL has 0 seats in the constituency: Ferrari is not elected. CR has 2 seats: quotient $\lfloor 154{,}877/2 \rfloor = 77{,}438$; quotients of attribution FdI $94{,}823/77{,}438 = 1.2245$, Lega $43{,}109/77{,}438 = 0.5567$, FI $0.2188$; FdI takes the integer-part seat, and the residual goes to Lega (0.5567). Lega's college list in P01 is: 1.~Diego Binelli; 2.~Vanessa Cattoi (already elected in a single-member district, hence skipped); 3.~Filippo Maturi. \emph{Diego Binelli} is elected. The FdI seat elects the same person under both interpretations.

\paragraph{Case study 2: Sardegna --- Romina Mura (PD) in under A; the M5S bench survives by two different legal routes.}
Sardegna also has a single college (P01, 7 seats). \emph{Under C}, CL has 2 seats: quotient $\lfloor(128{,}596+34{,}675)/2\rfloor = 81{,}635$; PD $128{,}596/81{,}635 = 1.5753$, AVS $0.4248$; integer parts give PD 1, and the residual also goes to PD (0.5753 $>$ 0.4248), but PD ends the national sweep 4 seats above its entitlement of 57, and Sardegna (decimal 0.5753) is the second-smallest of its four residual gains, so the compensation returns it to AVS (Table~\ref{tab:comp}). Final: PD 1 (its list opens with Silvio Lai, elected), AVS 1. \emph{Under A}, CL has 3 seats: quotient $\lfloor 163{,}271/3 \rfloor = 54{,}423$; PD $128{,}596/54{,}423 = 2.3629$, AVS $34{,}675/54{,}423 = 0.6371$; PD takes 2 integer-part seats, AVS the residual. PD's national surplus under A is only $+2$ (Veneto~2 and Campania~2 are returned; Sardegna's PD seats are integer-part seats and cannot be touched by the compensation, which by the letter of the law removes seats obtained \emph{with the decimal parts}). Final: PD 2 (Silvio Lai and, second on the list, \emph{Romina Mura}, elected) and AVS 1 (same person as under C).

For M5S the constituency seat count drops from 2 (C) to 1 (A), yet, remarkably, the M5S bench ends up identical in persons \emph{and} colleges, through two entirely different legal routes; the case is worth following because it shows the recovery chain of Art.~84 responding to the ambiguity. In both interpretations M5S cannot fill two of its seats in Campania 1 -- P02, and Art.~84(4) sends them to the constituencies with the largest \emph{unused} constituency decimal parts, one each. \emph{Under C}, the ranking is headed by Lombardia 2 (0.6348, returned in compensation) and Lombardia 4 (0.5685): the first recovered seat identifies the Lombardia 2 -- P01 list leader Alessandra Todde \emph{although she is already elected in Sardinia}; the multiple election is resolved per Art.~85 in favour of Lombardia 2 (list percentage 7.6198\% $<$ 21.8028\%), and 3.~Susanna Cherchi takes over the vacated Sardinian seat next to 2.~Emiliano Fenu (this is precisely the sequence recorded in the official minutes); the second seat elects Valentina Barzotti (Lombardia 4 -- P01). \emph{Under A}, M5S holds the Lombardia 2 seat as an ordinary seat (see below): Todde wins Lombardia 2 and Sardinia directly and Art.~85 sends her to Lombardia 2 up front; Sardinia's single ordinary seat goes to Fenu; and the recovery ranking, in which Lombardia 2 is now \emph{used} while Sardegna's 0.6781 is unused (the dynamic exclusion blocked it), sends the first Campania seat to Sardinia (electing Cherchi there) and the second, as before, to Barzotti. Same four deputies, same four colleges; but under C two of them sit by virtue of Art.~84 and Art.~85, under A by ordinary proclamation. The ambiguity of \arth{} is absorbed, in this instance, by the recovery machinery; one constituency over, in Lombardia 2, it is not (case study 3).

\paragraph{Case study 3: Lombardia 2 --- Umberto Bossi (Lega) and Devis Dori (AVS) out, Roberto Rampi (PD) in.}
Lombardia 2 has two colleges: P01 (4 seats, college quotient $\lfloor 397{,}358/4\rfloor = 99{,}339$, computed on the admitted lists' tallies, per Section~\ref{sec:validation}) and P02 (5 seats, $\lfloor 575{,}582/5\rfloor = 115{,}116$). At \arth{} level, C assigns CR 6 seats (5 direct $+$ 1 from the M5S compensation) and M5S 0; A assigns CR 5 and M5S 1.

\emph{Centre-right, under C} (6 seats): quotient $\lfloor(326{,}768 + 168{,}427 + 89{,}586)/6\rfloor = 97{,}463$; quotients of attribution FdI $3.3527$, Lega $1.7281$, FI $0.9192$; integer parts FdI 3, Lega 1; the two residuals go to FI (0.9192) and Lega (0.7281). Final: FdI 3, \textbf{Lega 2}, FI 1. At Art.~83-\emph{bis}, Lega's college decimal parts are 0.6566 in P01 (tally 65{,}225, college quotient 99{,}339) and 0.8965 in P02 (103{,}202, quotient 115{,}116), and Lega wins the college-level residual races in \emph{both} colleges, covering its two seats. Its P01 list is: 1.~\emph{Umberto Bossi}; 2.~Simona Bordonali (elected in a district, skipped); 3.~Matteo Luigi Bianchi. \textbf{Umberto Bossi is elected.} \emph{Under A} (5 seats): quotient $\lfloor 584{,}781/5 \rfloor = 116{,}956$; FdI $2.7939$, Lega $1.4401$, FI $0.7660$; integer parts FdI 2, Lega 1; residuals to FdI (0.7939) and FI (0.7660). Final: FdI 3, \textbf{Lega 1}, FI 1. At college level Lega still wins both college races, so it is now one seat \emph{above} its constituency entitlement, and the compensation of Art.~83-\emph{bis} takes back the seat gained with the smaller decimal part, P01 (0.6566 $<$ 0.8965), handing it to the deficit list M5S in \emph{its} best unused college, P01 (0.3398 $>$ 0.3255): the seat is won ordinarily by Alessandra Todde (case study 2), and \textbf{the founder of the Lega is not elected}. The FdI and FI totals are unchanged, and elect the same persons.

\emph{Centre-left} (2 seats under both): quotient $\lfloor(172{,}356+36{,}917)/2\rfloor = 104{,}636$; PD $1.6472$, AVS $0.3528$; the sweep gives PD 2, AVS 0 in both interpretations. \emph{Under C}, PD's fourth surplus seat is returned exactly here (decimal 0.6472; Table~\ref{tab:comp}), giving PD 1 and AVS 1. At college level PD wins the residual races in both P01 (0.7224) and P02 (0.8739); being one seat above its (reduced) entitlement, it gives up the smaller-decimal college P01, and AVS receives the seat in its own best unused college (P02, where its decimal 0.1889 beats its P01 decimal 0.1528), electing the P02 list leader \emph{Devis Dori}. \emph{Under A}, PD's surplus is exhausted elsewhere and Lombardia 2 keeps PD 2, AVS 0: PD covers both colleges, its P01 list leader \emph{Roberto Rampi} is elected, and Dori is not. (Note how the statutory rule that donor and recipient colleges are selected \emph{independently} is what sends AVS to P02 rather than to the college PD vacated; an earlier revision of our implementation that gave same-college priority here elected the P01 list leader instead: a one-sentence misreading worth one seat, caught by the name-level validation of Section~\ref{sec:validation}.)

\paragraph{Case study 4: Sicilia 2 --- the within-coalition sweep is itself order-dependent.}
At \arth{} level, Sicilia 2 is identical under A and C (CR 5, CL 2, M5S 3, A--IV 1). The difference arises one level down, inside \arti{}, which under interpretation A is also executed sequentially with dynamic exclusion. The CL quotient in Sicilia 2 (2 seats) is $\lfloor(127{,}820+22{,}483)/2\rfloor = 75{,}151$: PD $1.7008$, AVS $0.2992$. \emph{Under C}, PD takes the integer-part seat and the residual (0.7008 $>$ 0.2992): PD 2, AVS 0; the compensation removes PD seats elsewhere (its Sicilian decimal 0.7008 is too large to be selected). \emph{Under A}, by the time the within-CL sweep reaches Sicilia 2 (25th constituency), PD has already reached its national total of 57: it is dynamically excluded, and the residual goes to AVS (0.2992).

At Art.~83-\emph{bis} level the consequences unfold as follows. PD's college decimal parts are 0.5353 (P01, college quotient 78{,}841), 0.4607 (P02, quotient 86{,}702) and 0.5952 (P03, quotient 76{,}729), and PD wins the college residual race in all three colleges; Forza Italia wins in P02 (0.5116) and P03 (0.4588), one more than its entitlement of 1. \emph{Under C} (PD entitlement 2), PD gives up its smallest-decimal college, P02 (0.4607), where Azione--Italia Viva receives (its best unused decimal, 0.2410, electing Giuseppe Castiglione); FI gives up P03, and Lega receives in its own best unused college, P02 (0.2738); PD keeps P01 (electing the list leader \emph{Maria Stefania Marino}) and P03. This is, name for name and college for college, the officially proclaimed outcome. \emph{Under A} (PD entitlement 1), PD must give up P01 as well: Marino is not elected, and the seat reaches the remaining deficit list AVS, which receives it in \emph{its} best unused college (P02, decimal 0.1114 against 0.0981 in P03), electing the P02 list leader \emph{Pierpaolo Montalto}. The \emph{existence} of the AVS-for-PD swap in Sicilia 2 rests on the saturation event above and is noise-proof; the college placements ride on margins of $0.013$--$0.075$, one to two orders of magnitude above the attainable perturbations (Table~\ref{tab:margins}).

\subsection{The order of the constituencies under interpretation A}
\label{sec:names-Arev}

\begin{table}[t]
\centering\small
\caption{Named differences between the two runs of interpretation A: official constituency order vs.\ reversed order. Identical votes, identical algorithm; only the processing order changes.}
\label{tab:names-Arev}
\footnotesize
\begin{tabular}{llll}
\toprule
& Deputy & List & College \\
\midrule
\multirow{6}{*}{\shortstack[l]{Elected under the\\ official order only}}
& Maria Cecilia Guerra & PD & Piemonte 1 -- P01 \\
& Catia Polidori & FI & Umbria -- P01 \\
& Aldo Mattia & FdI & Basilicata -- P01 \\
& Romina Mura & PD & Sardegna -- P01 \\
& Pierpaolo Montalto & AVS & Sicilia 2 -- P02 \\
& Diego Binelli & Lega & Trentino-A.A. -- P01 \\
\midrule
\multirow{6}{*}{\shortstack[l]{Elected under the\\ reversed order only}}
& Giovanni Ravalli & FdI & Piemonte 1 -- P02 \\
& Paolo Nicol\`o Romano & AVS & Piemonte 2 -- P02 \\
& Valeria Alessandrini & Lega & Umbria -- P01 \\
& Michele Casino & FI & Basilicata -- P01 \\
& Maria Stefania Marino & PD & Sicilia 2 -- P01 \\
& Sara Ferrari & PD & Trentino-A.A. -- P01 \\
\midrule
Displaced
& Chiara Appendino & M5S & Piemonte 2 -- P02 $\to$ Piemonte 2 -- P01 \\
\bottomrule
\end{tabular}
\end{table}

Table~\ref{tab:names-Arev} lists the 12 deputies whose presence in the Chamber depends only on the order in which the constituencies are processed, plus one who changes college. The mechanism mirrors Section~\ref{sec:names-AC}: the dynamic exclusion bites in the \emph{last} constituencies of whatever order is chosen. Under the official order the divergences are in Sardegna, Trentino and (via \arti{}) Sicilia 2, so the reversed run, which processes those constituencies \emph{first}, agrees there with C, and the six ``official-order'' names of Table~\ref{tab:names-Arev} include three of the four A-only names of Table~\ref{tab:names-AC} (Binelli, Mura, Montalto); the fourth, Rampi, happens to be elected under both orders. Under the reversed order the divergences appear at the opposite end of the list:

\begin{itemize}
\item \textbf{Lombardia 1} (25th processed): M5S is saturated (41/41) and skipped (decimal 0.5189); the residual goes to CL with 0.2467.
\item \textbf{Piemonte 2} (26th): M5S skipped again (0.7805); CL takes the seat with 0.3426.
\item \textbf{Piemonte 1} (27th): CL is now itself \emph{above} its entitlement (69/68) and skipped (0.4179); M5S skipped (0.3516); CR takes the seat with 0.2505.
\end{itemize}
After the sweep CL is at $+1$ and M5S at $+1$, CR at $-2$; the compensation (Table~\ref{tab:comp}) takes back CL's just-gained Lombardia 1 seat (0.2467, its only residual gain small enough) and M5S's Lombardia 3 seat (0.5532), both to CR. The final unit-level picture (Appendix~\ref{app:assh}): relative to C, the reversed-order run moves one seat in Piemonte 1 (CR 5 / CL 3 instead of 4/4), one in Piemonte 2 (CL 3 / M5S 0 instead of 2/1), and one in Lombardia 2 (CR 5 / M5S 1 instead of 6/0, as in the official-order run).

\paragraph{Case study 5: Piemonte --- Appendino changes college, Iaria survives by two routes, Guerra loses her seat.}
Under the official order (where A agrees with C in Piedmont): M5S holds one seat in Piemonte 1 (college P01) and one in Piemonte 2, which the Art.~83-\emph{bis} compensation places in college P02 (M5S receives on its largest unused college decimal: 0.4494 in P02 vs 0.3621 in P01). Chiara Appendino heads the M5S lists in all four Piedmontese colleges and wins in both Piemonte 1 -- P01 and Piemonte 2 -- P02; Art.~85 sends her where the list percentage is lowest (8.0679\% vs 12.96\%), i.e.\ to Piemonte 2 -- P02, and the Piemonte 1 -- P01 seat goes to the second name, \emph{Antonino Iaria}. The CL in Piemonte 1 has 4 seats: quotient $\lfloor(240{,}909+53{,}840)/4\rfloor = 73{,}687$, PD $3.2694$, AVS $0.7307$: PD 3, AVS 1 (residual). PD's three seats fill P01 twice and P02 once; on the P01 list, Debora Serracchiani (multi-candidate, proclaimed in Friuli-Venezia Giulia at 18.40\%) is skipped, electing Mauro Laus and \emph{Maria Cecilia Guerra}.

Under the reversed order: M5S loses its Piemonte 2 seat (taken by CL at letter (h), Section above). CL in Piemonte 2 now has 3 seats: quotient $\lfloor 209{,}861/3 \rfloor = 69{,}953$, PD $2.5521$, AVS $0.4479$: PD 2 by integer parts, and the residual goes to AVS, since PD is dynamically excluded by then (the within-coalition sweep also runs in reversed order, and PD reaches 57 before Piedmont is processed). AVS's new Piemonte 2 seat is placed by the college compensation on its largest unused college decimal, P02 (0.1680) over P01 (0.1429), electing the P02 list leader \emph{Paolo Nicol\`o Romano}. In Piemonte 1, CL drops to 3 seats: quotient $\lfloor 294{,}749/3 \rfloor = 98{,}249$, PD $2.4520$, AVS $0.5480$: PD 2 (one per college; \emph{Guerra loses her seat}), AVS 1. CR in Piemonte 1 rises to 5 seats: quotient $\lfloor 428{,}674/5\rfloor = 85{,}734$, FdI $3.0409$, Lega $1.0886$, FI $0.8706$: FdI 3, Lega 1, FI 1 (residual), one FdI seat more than under the official order (FdI $2.4327$, residuals to Lega 0.8708 and FI 0.6965). FdI's extra seat goes to P02, where, after Augusta Montaruli (district winner, skipped), Marco Perissa (proclaimed in P01) and Immacolata Zurzolo (elected on the list's other P02 seat), the next available name is 4.~\emph{Giovanni Ravalli}, elected. Finally, the Art.~84(4) chain of case study 2 re-routes once more: under the reversed order the largest unused M5S constituency decimal is Piemonte 2 itself (0.7805, blocked by the dynamic exclusion at the end of the reversed sweep), so the first Campania seat lands in Piemonte 2 -- P01 and identifies \emph{Appendino herself}, already elected in Piemonte 1 -- P01; the post-recovery Art.~85 check proclaims her in Piemonte 2 -- P01 (8.1702\% $<$ 12.96\%), and \emph{Iaria re-enters as her substitute} in Piemonte 1 -- P01. Iaria thus sits under both orders (as an ordinary winner in one, as an Art.~85 substitute in the other) while Appendino represents college P02 or P01 of Piemonte 2 depending on nothing but the processing order.

\paragraph{Case study 6: Umbria and Basilicata --- the compensation changes its beneficiary.}
Under both A runs, FdI ends the within-CR sweep above its national total of 69 and must return seats; but the size of its surplus differs (one seat under the official order, two under the reversed order, the second coming from the Piemonte 1 integer part gained above). Under the official order the surplus seat is returned in Umbria (FdI's smallest residual-gain decimal, 0.3577), and the receiving list is the \emph{only} deficit list at that point, FI (Umbria decimals: FI 0.3010, Lega 0.3413): FI's seat elects its Umbrian list leader \emph{Catia Polidori}. Under the reversed order both FI and Lega are in deficit; in Umbria the law gives the seat to the deficit list with the \emph{largest} unused decimal part, Lega (0.3413 $>$ 0.3010), electing \emph{Valeria Alessandrini}; the second surplus seat is returned in Basilicata (decimal 0.4976), where the remaining deficit list FI (0.2564) receives it, electing \emph{Michele Casino}; and FdI's own Basilicata seat, which under the official order elects \emph{Aldo Mattia}, disappears. In Sardegna, finally, the reversed run agrees with C at unit level (CL 2, M5S 2), so Romina Mura is not elected; and since M5S also holds Lombardia 2 under this run, Todde again leaves Sardinia (Art.~85, 7.6198\% $<$ 21.8028\%), so the two regular M5S seats go to 2.~Emiliano Fenu and 3.~\emph{Susanna Cherchi}. Cherchi thus sits in the Chamber under both orders, but under the official order as an Art.~84(4) replacement for a Campanian seat (case study 2), under the reversed order as an ordinary Sardinian proclamation: the same person, elected from the same list position, on two different legal bases.

\subsection{Interpretation B: where two seats remain unassigned}
\label{sec:names-B}

Under interpretation B (all integer parts first, then residual seats with dynamic exclusion, official order), the sweep proceeds like C until M5S reaches its national total of 41 at Campania 1. From there:
\begin{itemize}
\item \textbf{Basilicata}: M5S, though holding the largest decimal part (0.8156), is excluded; the two residual seats go to CL (0.6753) and A--IV (0.3186).
\item \textbf{Calabria}: M5S (0.5789) excluded; seats to CR (0.5424) and CL (0.5150), bringing CL to 67; CL completes its total of 68 with the Sicilia 1 residual at the following step.
\item \textbf{Sicilia 2} (2 residual seats): the first goes to CR (0.8613). For the second, \emph{every} admitted unit is unavailable: A--IV saturated (0.6212), M5S saturated (0.4793), CL saturated (0.0382), and CR has already used its decimal part in this constituency. \textbf{The seat cannot be assigned.}
\item \textbf{Sardegna}: same situation, with CR taking the first residual (0.9611); for the second, M5S (0.6781), A--IV (0.3547) and CL (0.0062) are saturated. \textbf{The seat cannot be assigned.}
\item \textbf{Trentino-Alto Adige}: CL saturated (0.8562) is skipped; SVP (0.7649) and CR (0.0124) take the two seats.
\end{itemize}
The sweep ends with CR at 112 of 114 and \emph{no} unit above its total: there is nothing the compensation procedure can transfer, and within this mechanical execution of the text the two seats remain unassigned: absent the kind of extra-textual suppletive rule discussed in Section~\ref{sec:results}, the Chamber would convene with 398 of 400 members. The deficit propagates into \arti{} (FI $-1$, Lega $-1$) and below; the text provides no rule for any of this. For completeness, Table~\ref{tab:names-B} reports the named outcome of our mechanical continuation of the pipeline (which elects 390 of the 391 modelled deputies and, in the Sicilian college compensation, is forced into one further inconsistency: a Partito Democratico college seat in Sicilia 2 that exceeds the list's constituency entitlement and cannot be compensated away, so that the realised list totals of this completed B are Forza Italia $-1$, Lega $-1$ and PD $+1$ with respect to Table~\ref{tab:apportionment}). These names are indicative (unlike Sections~\ref{sec:names-AC} and~\ref{sec:names-Arev}, they depend on choices the law does not define), but the two lost seats, their location, and the impossibility of repair are exact consequences of interpretation B.

\begin{table}[t]
\centering\small
\caption{Named differences between interpretation B and interpretation C (official order). The B column reflects one mechanical completion of the pipeline on B's deficient seat totals (see text); B elects one deputy fewer overall (390 vs 391).}
\label{tab:names-B}
\footnotesize
\begin{tabular}{p{6.9cm}p{6.9cm}}
\toprule
Elected under C, \emph{not} under B (12) & Elected under B, \emph{not} under C (11) \\
\midrule
Umberto Bossi (Lega, Lombardia 2 -- P01);
Giulio Centemero (Lega, Lombardia 3 -- P01);
Grazia Di Maggio (FdI, Lombardia 1 -- P02);
Devis Dori (AVS, Lombardia 2 -- P02);
Valentina Barzotti (M5S, Lombardia 4 -- P01, via Art.~84(4));
Luana Zanella (AVS, Veneto 2 -- P02);
Francesco Mari (AVS, Campania 2 -- P02);
Elisa Scutell\`a (M5S, Calabria -- P01, via Art.~84(3));
Giovanna Iacono (PD, Sicilia 1 -- P02);
Giuseppe Castiglione (A--IV, Sicilia 2 -- P02);
Catia Polidori (FI, Umbria -- P01);
Sara Ferrari (PD, Trentino-A.A. -- P01)
&
Roberto Rampi (PD, Lombardia 2 -- P01);
Sara Montrasio (M5S, Lombardia 1 -- P02);
Samuel Sorial (M5S, Lombardia 3 -- P02);
Alessia Rotta (PD, Veneto 2 -- P03);
Rosa D'Amelio (PD, Campania 2 -- P02);
Mario Polese (A--IV, Basilicata -- P01);
Filippo Sestito (AVS, Calabria -- P01);
Marilena Grassadonia (AVS, Sicilia 1 -- P01);
Pierpaolo Montalto (AVS, Sicilia 2 -- P02);
Marco Squarta (FdI, Umbria -- P01);
Diego Binelli (Lega, Trentino-A.A. -- P01) \\
\midrule
\multicolumn{2}{p{14.2cm}}{Displaced: Giuseppe Conte (M5S, Lombardia 1 -- P02 $\to$ P01)} \\
\bottomrule
\end{tabular}
\end{table}

Several of the B-vs-C differences have transparent mechanisms. Because B never creates surpluses, the three M5S residual seats in Lombardia 1, 2 and 3 that C's compensation hands to the centre-right \emph{stay with M5S} (electing Montrasio and Sorial, and turning Todde's Lombardia 2 seat into an ordinary one), and the centre-right loses one seat in each Lombard constituency (Di Maggio, Bossi, Centemero). M5S's exclusion in Basilicata and Calabria hands one seat to A--IV (Polese) and, by removing the Calabrian exhausted-list case that under C is resolved through Art.~84(3), unseats Elisa Scutell\`a, the M5S deputy whose seat was the object of the real 2025 judicial recount; the freed CL seat in Calabria elects AVS's list leader Sestito. With M5S saturated earlier, the Art.~84 chain shortens by one seat (Barzotti out). The Trentino and Sicilia 2 endgames replicate the A-run mechanics (Binelli in, Ferrari out; Montalto in, at the expense of the A--IV placement that elects Castiglione under C).

\subsection{Robustness of the named results}
\label{sec:names-robust}

The input data agree with the official national tallies within $\pm0.06\%$ (Section~\ref{sec:validation}, Level 2), and the decimal parts computed from them deviate from the official ones by up to $\pm 0.001$ (Level 4), deviations that, on other votes, could flip the territorial placement of marginal seats \emph{within} a fixed interpretation. A reviewer of the named results of this section is therefore entitled to ask whether the differences \emph{between} interpretations (the deputies of Tables~\ref{tab:names-AC}--\ref{tab:names-B}) could be artefacts of such noise. They cannot, for a structural reason and a quantitative one, and we verified both.

\paragraph{Structure: the divergences are anchored to integer events.} A within-interpretation marginal seat is decided by a single close comparison of decimal parts. The A-vs-C divergences are not: each originates in a \emph{saturation event}, a unit's running total reaching an integer entitlement (M5S reaching $43 > 41$ before Sardegna; CL reaching $68$ before Trentino; PD reaching $57$ before Sicilia 2), and integer totals move only when some upstream decimal race flips. The named chains then run through pairwise races whose margins are large. Table~\ref{tab:margins} lists the decisive pairwise comparisons appearing in the case studies of Sections~\ref{sec:names-AC}--\ref{sec:names-Arev}, together with a conservative bound on how far the perturbation model below can move each decimal part: every margin exceeds its bound by a factor between 7 and 181.

\begin{table}[t]
\centering\small
\caption{Decisive pairwise comparisons in the named chains of Sections~\ref{sec:names-AC}--\ref{sec:names-Arev}: margin between the winning and losing decimal parts, and a conservative bound on the largest decimal-part shift attainable under the $\pm0.06\%$ perturbation model (bound $= (a_w+a_l)(4\varepsilon + 1/q)$ with $\varepsilon = 6\cdot10^{-4}$, $a$ the quotients of attribution, $q$ the relevant electoral quotient). Saturation-type events (a unit excluded because an integer entitlement is reached) have no decimal margin and are covered by the Monte Carlo analysis.}
\label{tab:margins}
\begin{tabular}{lrrr}
\toprule
Comparison & Margin & Bound & Ratio \\
\midrule
(h) comp.\ C: 3rd M5S return, Lombardia 2 vs Sardegna & 0.0433 & 0.0056 & 8$\times$ \\
(i) CR Lombardia 2 (A): residual, FI vs Lega & 0.3259 & 0.0053 & 61$\times$ \\
(i) CR Lombardia 2 (C): residual, Lega vs FdI & 0.3754 & 0.0122 & 31$\times$ \\
(i) CR Trentino-A.A.\ (A): residual, Lega vs FI & 0.3379 & 0.0019 & 181$\times$ \\
(i) CL Sardegna (A): residual, AVS vs PD & 0.2742 & 0.0073 & 38$\times$ \\
(i) CL Sicilia 2 (C): residual, PD vs AVS & 0.4016 & 0.0048 & 83$\times$ \\
(i) comp.\ C: 4th PD return, Lombardia 2 vs Sicilia 2 & 0.0536 & 0.0081 & 7$\times$ \\
(i) comp.\ A$_{\mathrm{rev}}$ receiver in Umbria: Lega vs FI & 0.0403 & 0.0015 & 26$\times$ \\
(i) CR Piemonte 1 (A$_{\mathrm{rev}}$): residual, FI vs Lega & 0.7820 & 0.0047 & 166$\times$ \\
83-\emph{bis} comp.\ (A): Lega Lombardia 2 cedes P01 vs P02 & 0.2399 & 0.0039 & 61$\times$ \\
83-\emph{bis} comp.\ (A): M5S Lombardia 2 receives P01 vs P02 & 0.0143 & 0.0016 & 9$\times$ \\
83-\emph{bis} comp.\ (C): PD Sicilia 2 cedes P02 vs P01 & 0.0746 & 0.0024 & 31$\times$ \\
83-\emph{bis} comp.\ (A): AVS Sicilia 2 receives P02 vs P03 & 0.0133 & 0.0005 & 27$\times$ \\
Art.~84(4) chain (C): Lombardia 2 vs Lombardia 4 & 0.0663 & 0.0029 & 23$\times$ \\
Art.~84(4) chain (A): Sardegna vs Lombardia 4 & 0.1096 & 0.0030 & 37$\times$ \\
Art.~85, Todde: Lombardia 2 vs Sardegna (percentage points) & 14.18 & $\sim$0.02 & $\gg$100$\times$ \\
\bottomrule
\end{tabular}
\end{table}

\paragraph{Monte Carlo perturbation.} We perturb every list tally in every multi-member college by two independent multiplicative factors, one per list (a national bias) and one per list and constituency (a territorial bias), rebuild all derived tallies (constituency, college totals, percentage tallies), and re-run the pipeline from Phase~2 to Phase~5 for both A and C on each perturbed input, comparing the resulting A-vs-C named difference set with the baseline set of Table~\ref{tab:names-AC}. District winners and candidate-level votes are left untouched: their reconstruction from the data is exact (Section~\ref{sec:data}) and no divergence chain involves them; the perturbations do, however, propagate to every quantity downstream of the list tallies (including the percentage tallies that drive Art.~85 and the decimal parts that drive Art.~84), so the multi-candidacy and recovery chains are fully exercised on each draw. Four noise models of increasing severity are used: factors drawn uniformly in $\pm0.06\%$ (the observed national deviation of the open data); factors at the extremes $\pm0.06\%$ (the most adversarial signs at the observed magnitude); and extreme factors at $\pm0.12\%$ and $\pm0.30\%$, i.e.\ two and five times the observed deviation, to account for the possibility that the true data error concentrates territorially, and is adversarial in sign, rather than spreading evenly as symmetric statistical scatter. By design these models therefore span more than idealised uniform i.i.d.\ noise, probing \emph{systematic} data bias rather than only random error. We stress the epistemic status of this design: it is a \emph{stress test}, not a calibrated probabilistic model of the certified data-generating process: the true residual error of the open data is unknown beyond its aggregate magnitude, so we probe families of deviations that dominate it in size and in adversarialness of sign, and draw only worst-case conclusions (of the form ``the difference set survives every draw at the observed noise level''), never distributional ones. One hundred draws per model suffice for conclusions of that form; we make no finer inferential claim. Table~\ref{tab:mc} reports the results (100 draws per model). In every draw of every model, the set of admitted units and the national apportionment of Table~\ref{tab:apportionment} are unchanged: the thresholds of Art.~83(1) are never approached at these noise levels.

\begin{table}[t]
\centering\small
\caption{Monte Carlo robustness of the A-vs-C named differences (100 draws per noise model). ``Set identical'': draws in which the A-vs-C named difference set equals the baseline of Table~\ref{tab:names-AC} exactly. ``Min.\ persistence'': the per-name minimum, over the 8 baseline differences, of the fraction of draws in which that difference persists. ``Extra diffs'': mean number of non-baseline differences per draw. The last two columns measure the \emph{power} of the test: how often, and how much, the perturbation changes the composition of the C parliament itself relative to unperturbed C.}
\label{tab:mc}
\footnotesize
\begin{tabular}{lrrrrr}
\toprule
Noise model & Set identical & Min.\ persistence & Extra diffs & C changed & Names changed in C \\
\midrule
uniform $\pm0.06\%$ & 100/100 & 100\% & 0.00 & 0/100 & 0.0 (max 0) \\
extreme $\pm0.06\%$ & 100/100 & 100\% & 0.00 & 4/100 & 0.1 (max 2) \\
extreme $\pm0.12\%$ & 86/100 & 86\% & 1.26 & 33/100 & 1.9 (max 13) \\
extreme $\pm0.30\%$ & 47/100 & 50\% & 3.85 & 62/100 & 4.6 (max 19) \\
\bottomrule
\end{tabular}
\end{table}

Three facts stand out. First, at the noise level actually exhibited by the data ($\pm0.06\%$, whether uniform or adversarial in sign), \emph{all eight named differences persist in every draw, with no additions}, and the national apportionment of Table~\ref{tab:apportionment} never changes; the test has power already at this level, since the adversarial-sign model does flip names \emph{within} the C parliament in 4\% of draws (up to 2 names, the class of behaviour underlying the residual official-data mismatches), without ever touching the difference set. Second, doubling the noise scrambles the C parliament more substantially (33\% of draws, up to 13 names), yet the A-vs-C difference set survives intact in 86\% of draws and each individual difference in at least 86\%. Third, at five times the observed noise the composition of the C parliament changes in 62\% of draws (up to 19 names), while the difference set still persists intact in 47\% of draws and each individual difference in at least 50\%: the differences between interpretations are markedly \emph{more} stable than the marginal seats within an interpretation, as the anchoring argument predicts. We conclude that the identities in Tables~\ref{tab:names-AC} and~\ref{tab:names-Arev} are not artefacts of input noise; for Table~\ref{tab:names-B} the same holds for the seat-level facts (which units lose seats and where), while the B name list additionally depends on the undefined downstream behaviour already discussed in Section~\ref{sec:names-B}.

\subsection{Summary}

Across the three experiments, 31 distinct persons enter or leave the Chamber, and 2 more change the college they represent, without a single vote changing, and without any change in the party totals, except for the two seats that interpretation B leaves unassigned. Every one of these changes traces back, through a handful of exactly reproducible divisions, to the single ambiguous sentence of \arth{} identified by Crafa.

\section{Discussion and limitations}
\label{sec:discussion}

\paragraph{What the experiment establishes.} The three interpretations identified by Crafa are not merely abstractly distinct: on the actual votes of the 2022 election they produce materially different parliaments, and the differences have names and faces (Section~\ref{sec:names}). At the same time, the compensation machinery of the current law is effective at protecting party-level proportionality: all completed variants agree exactly on every party's seat count. The normative implication is the one drawn in~\cite{crafa2023}: the text should be amended to prescribe one algorithm unambiguously (the minimal amendment proposed there corresponds to Algorithm~C, which our experiments show to be order-independent on real data), and the compensation procedure should be fully specified, ideally in algorithmic form.

\paragraph{An inventory of text--practice gaps.} Beyond the central ambiguity of \arth{}, the reconstruction surfaced every point at which the executed procedure \emph{specifies} something the enacted text leaves open or says differently; we collect them here because they are findings in their own right, all of the same kind as Crafa's. (1)~\emph{\arth{}, exclusion clause}: the most immediate sequential reading (Algorithm~A) is not what is executed; practice computes the exclusion once, statically (Algorithm~C), the subject of this paper. (2)~\emph{Art.~83-bis(1), college quotient}: the text divides ``the sum of the college tallies of \emph{all} the lists''; the decimal parts recorded in the official minutes prove that the sum is taken over the admitted lists only (Section~\ref{sec:validation}, Level~4). (3)~\emph{Art.~84 recoveries}: the text does not say whether the candidate receiving a recovered seat may be someone already elected elsewhere, nor that each decimal part supports a single recovered seat; practice identifies the candidate regardless, resolves the resulting multiple election afterwards under Art.~85, and moves down the two-tier (unused-first) ranking one seat per decimal part. (4)~\emph{Arts.~83-bis--85, seat and person placement}: which college a multiply-listed candidate vacates is decided by percentage tallies that, in Toscana, the returning office itself computed in two mutually inconsistent ways within a single \emph{verbale} (Gori vs.\ Michelotti); the AVS multiple-election chain that decides between Paglia and Piccolotti originates one step earlier, in an Art.~83-\emph{bis} college attribution resting on a very small decimal margin. These are our two residual disagreements, and both are, since July~2025, under formal investigation by the Chamber's Committee on Elections, together with two further seats raising questions of the same family~\cite{giunta2025}. (5)~\emph{Historically}, the stranded seats produced by the Mattarella-law procedure in 1994--2001 were themselves recovered by a practice, not by the text~\cite[p.~66]{camera2019}. None of these gaps is visible in aggregate seat totals; every one of them moved, or could move, specific named seats; and none of them is resolvable by reading the statute alone, which is precisely the argument for writing the algorithm into the law; Crafa's published analysis takes exactly this step, proposing a concrete draft amendment of \arth{}~\cite[App.~2]{crafa2024}.

\paragraph{The arbitrariness is not marginal.} Consider \emph{how} the sequential interpretation allocates its final seats. In the Sardinian endgame of Section~\ref{sec:names-AC}, a seat is assigned to a decimal part of 0.0062 while a decimal part of 0.6781 is passed over; in Trentino, 0.0124 defeats 0.8562. The law's own selection criterion (largest decimal parts) is inverted by the interplay of the exclusion clause with the processing order, an interplay the legislator manifestly did not design, since the text does not even fix the order. The same mechanism, run in the opposite order, produces the same kind of inversion at the other end of the country (0.2467 and 0.2505 winning in Lombardy and Piedmont). Which of Guerra and Ravalli, Polidori and Alessandrini, Mattia and Casino, Mura and Marino sits in the Chamber is decided by nothing in the voters' expression and nothing in the law's stated criteria.

\paragraph{Fidelity limits.} Three caveats bound the fidelity claim of Section~\ref{sec:validation}. (i) The municipal open-data file is near-definitive, not certified-definitive: the decimal parts computed from it deviate from those recorded in the official minutes by up to $\pm0.001$. On the 2022 data this noise never crosses a decision margin in the transcribed minutes, and both residual name disagreements lie in multiple-election chains whose handling is itself under parliamentary investigation (Section~\ref{sec:validation}); on different votes, however, marginal placements within an interpretation could flip. The variant \emph{comparisons} of Sections~\ref{sec:results} and~\ref{sec:names} are, however, robust to this caveat: all variants run on identical inputs, the divergence mechanisms are anchored to integer-valued saturation events rather than to close decimal races, and the Monte Carlo analysis of Section~\ref{sec:names-robust} shows that the full set of named differences survives perturbations at, and well beyond, the observed noise level. (ii) The Aosta Valley and overseas constituencies are out of scope (9 seats, elected under different rules, not interacting with the pipeline under study). (iii) Paragraphs 5 and 7 of Art.~84 and the statutory draws are implemented as logged stubs, since no 2022 case exercises them; a different election could require completing them.

\paragraph{Robustness of the headline numbers.} The counts of Table~\ref{tab:variants} are specific to the 2022 votes, but no longer to a single permutation: the order study of Section~\ref{sec:results} brackets the order effect of A (0--13 entering and 0--14 leaving against practice; 161 distinct outcomes over 200 random orders, 560 over 1{,}000, and the order-contingent set closing at 140 persons over a near-exhaustive 15 million, with stable summary statistics; 29\% of orders incomplete and a further 9\% with altered party totals), the reversal experiment falling near the middle of that distribution, and the Monte Carlo analysis of Section~\ref{sec:names-robust} certifies the named A-vs-C differences against input noise. The qualitative findings (A order-dependent and lossy under a third of the sampled orders, C order-independent (Proposition~\ref{prop:C}), B lossy under the documented order and under nearly half of the sampled orders, party totals invariant under the documented executions) are each mechanistically explained by the structure of the respective algorithm.

\paragraph{What is established, and how strongly.} We state the epistemic status of the individual claims. (i)~\emph{Established computationally, beyond the reach of input noise}: the existence and mechanism of the ambiguity; the exact national apportionment; the unit- and list-level divergences among A, B and C and their named consequences, and the stranding and party-migration counts of the order sweep (Sections~\ref{sec:results}--\ref{sec:names}, robustness in Section~\ref{sec:names-robust}). (ii)~\emph{Established on the modelled data, but sensitive in principle to input revisions}: the identity of marginal seats \emph{within} a fixed interpretation: the two residual disagreements of Section~\ref{sec:validation}, and the within-C flips seen in the Monte Carlo power analysis of Section~\ref{sec:names-robust}, are of this kind; the B name list of Table~\ref{tab:names-B} shares this weaker status for the additional reason discussed there. (iii)~\emph{Informed argument, not demonstration}: the legal-interpretive standing of the three algorithms (Section~\ref{sec:legal-scope}) and the constitutional significance of the phenomenon, on which we defer to the legal literature.

\section{Reproducibility}
\label{sec:repro}

Every number in this article is either printed in it (Tables~\ref{tab:apportionment}--\ref{tab:mc} and Appendices~\ref{app:letth}--\ref{app:i}) or derivable from those by the statutory procedure described in Sections~\ref{sec:law}--\ref{sec:impl}; in particular, the appendix tables were generated programmatically from the implementation's intermediate state, not transcribed by hand. The implementation (\texttt{rosatellum.py}, $\sim$1{,}170 lines of documented Python), the scripts and per-order/per-person outputs of the order studies of Section~\ref{sec:results} (200 and 1{,}000 random orders, per-deputy election frequencies with colleges, completeness and party-total scan, and the factored near-exhaustive sweep of 15 million orders with its consistency test) and the Monte Carlo scripts of Section~\ref{sec:names-robust}, the derived input table of per-college seats, the outputs of all runs (\texttt{eletti\_v2\_A/B/C.csv}), the official-proclamation dataset extracted from \texttt{dati.camera.it}, the per-MP comparison table and the full audit logs are publicly archived on Zenodo at \texttt{doi:10.5281/zenodo.21219963}, together with the OCR transcriptions of the UECN minutes used for the Level-4 validation (development repository: \texttt{github.com/coppolapaolo/one-vote-several-parliaments}). A full run of the pipeline (all phases, one variant) takes under a minute on commodity hardware; Phase~1 (the municipal-level aggregation) dominates and can be shared across variant runs. The random seed only affects statutory draws, none of which is triggered by the 2022 data.

The exact versions of the inputs and of the code are identified by the following SHA-256 digests; the Ministry of the Interior files were retrieved from the Eligendo open-data platform~\cite{eligendo} in 2023, and the official proclamations from \texttt{dati.camera.it}~\cite{daticamera} via SPARQL in 2026 (mandates of the XIX legislature with start date October 2022).

\begin{center}\scriptsize
\begin{tabular}{ll}
\toprule
File & SHA-256 \\
\midrule
\texttt{Camera\_Italia\_LivComune.txt} & \texttt{9f4f49dd0dfde3d1ea32ccda9665d6d230959f97ce074418ef2553242aa675b3} \\
\texttt{CAMERA\_ITALIA\_20220925\_uni.csv} & \texttt{c0fea844a6bf1b10896aae80351b5ffeb637ef918cb96e223bfcad783b5dfcf3} \\
\texttt{CAMERA\_ITALIA\_20220925\_pluri.csv} & \texttt{0e8b5c3af90c0f0fcf4f68421b7ac176eb570a1b9b75e123f86c77e5ce02c95e} \\
\texttt{spettanti\_circoscrizioni.csv} & \texttt{de7dd1c4d96984b7f364540a83cdf94a06353631295ac2b7cd5c6c45a8379ce4} \\
\texttt{spettanti\_collegi\_plurinominali.csv} & \texttt{f4a55356cfc3fdf18f0da18efbbaf89413efd3fd54d873a89a209be1d93ebd1a} \\
\texttt{eletti\_ufficiali\_2022.csv} & \texttt{f4b3a84769bf294b3f7b257401f42f12b3c71a78d9f954e19544d786bb826b94} \\
\texttt{rosatellum.py} & \texttt{6ce07712d09f95d4676dd3d4c45f4866e804a4b38412dd08413a67530a5316c6} \\
\bottomrule
\end{tabular}
\end{center}

\section{Conclusion}
\label{sec:conclusion}

We turned Crafa's algorithmic reading of the Italian electoral law into an executable, validated model, and ran the algorithmically admissible interpretations of the text on the real votes of the 2022 general election. The ambiguity of \arth{} is real and quantitatively significant at the level that matters for representation (who sits in parliament and for which territory): thirty-one people enter or leave the Chamber across the tested readings of one sentence of the law, among them the founder of one of the governing parties, and the differences are robust to input noise well beyond the residual uncertainty of the data. Under the executions documented in practice the ambiguity is invisible at the level of party totals, which the compensation procedure protects, yet the Mattarella-style reading strands two seats that the text provides no way to assign, and the most immediate sequential reading, run under random processing orders, strands seats in 29\% of cases and in a further 9\% silently moves one between parties; across the 1{,}000 sampled orders, the seats of 119 named individuals are in play. A parliament of the Italian Republic elected under the current text is, in a precise technical sense, one of several parliaments compatible with the same votes; which one, depends on implementation choices that the law does not make. This is also the general lesson of the exercise for computational law: legislative ambiguity is not only a matter of contested meanings, but of unspecified execution, and executing the text is how one finds it, and how one measures it. The fix is known and cheap: write the algorithm into the law. Our implementation shows that the operational interpretation (Algorithm~C) is a sound candidate (order-independent on real data and exactly compatible with the official 2022 outcome), and provides a reference against which any legislative reformulation can be tested, vote by vote and name by name.

\subsection*{Acknowledgements}
The author thanks Silvia Crafa, whose analysis of the electoral law motivated this work.

\subsection*{Use of generative AI and AI-assisted technologies}
In preparing this work the author made extensive and pervasive use of generative AI (large language model--based coding and writing assistants), employed interactively as a working tool throughout the research: in developing and debugging the software that implements the electoral pipeline (Section~\ref{sec:impl}), in extracting and processing the input data (Section~\ref{sec:data}), in cross-checking intermediate results, and in drafting and revising the text of the manuscript. The tools were used under continuous human direction and supervision: the author set every objective, questioned, corrected and iterated on all outputs, made all substantive decisions, and independently verified every quantitative result against the official public records (Section~\ref{sec:validation}). All code and derived data are openly archived and independently re-executable (Section~\ref{sec:repro}). The author takes full responsibility for the entire content of the manuscript, including any part produced with the assistance of these tools.

\appendix

\section{The complete Art.~83(1)(h) computation}
\label{app:letth}

Table~\ref{tab:letth-full} publishes, for each of the 27 constituencies in the official order used throughout the paper, the full inputs and intermediate values of the territorial distribution: the proportional seats $s_c$, the sum of the constituency tallies of the admitted units, the constituency electoral quotient $q_c$ (truncated), the number of residual seats $r_c$, and each unit's tally, quotient of attribution, integer part and decimal part. These values are shared by all interpretations; A, B and C differ only in how the $r_c$ residual seats are assigned (Algorithms~\ref{alg:A}--\ref{alg:C}). Summing the integer parts gives 190 seats (CR 100, CL 54, M5S 30, A--IV 6, SVP 0) and $\sum_c r_c = 55$ residual seats.

% tabella generata automaticamente da genera_tabelle.py
\begin{small}
\begin{longtable}{@{}llrrr@{\hspace{1.2em}}lrrrr@{}}
\caption{Art.~83(1)(h), complete computation for all 27 constituencies: proportional
seats $s_c$, total constituency tallies of the admitted units, constituency electoral
quotient $q_c=\lfloor\text{total}/s_c\rfloor$, number of residual seats $r_c$ (seats minus
sum of integer parts), and, for each unit, its constituency tally, quotient of attribution
$\mathit{cifra}/q_c$, integer part and decimal part. Units: CR = centre-right coalition,
CL = centre-left coalition, M5S = Movimento 5 Stelle, A--IV = Azione--Italia Viva,
SVP = SVP--PATT.}
\label{tab:letth-full}\\
\toprule
Constituency & $s_c$ & Total tallies & $q_c$ & $r_c$ & Unit & Tally & Quotient & Int. & Dec. \\
\midrule
\endfirsthead
\multicolumn{10}{l}{\small\emph{Table \ref{tab:letth-full} (continued)}}\\
\toprule
Constituency & $s_c$ & Total tallies & $q_c$ & $r_c$ & Unit & Tally & Quotient & Int. & Dec. \\
\midrule
\endhead
\bottomrule
\endlastfoot
\multirow{4}{*}{Piemonte 1} & \multirow{4}{*}{10} & \multirow{4}{*}{1\,008\,523} & \multirow{4}{*}{100\,852} & \multirow{4}{*}{2} & CR & 428\,674 & 4.2505 & 4 & 0.2505 \\
 & & & & & CL & 344\,705 & 3.4179 & 3 & 0.4179 \\
 & & & & & M5S & 136\,313 & 1.3516 & 1 & 0.3516 \\
 & & & & & A--IV & 98\,831 & 0.9800 & 0 & 0.9800 \\
\midrule
\multirow{4}{*}{Piemonte 2} & \multirow{4}{*}{9} & \multirow{4}{*}{944\,029} & \multirow{4}{*}{104\,892} & \multirow{4}{*}{2} & CR & 529\,113 & 5.0444 & 5 & 0.0444 \\
 & & & & & CL & 245\,720 & 2.3426 & 2 & 0.3426 \\
 & & & & & M5S & 81\,873 & 0.7805 & 0 & 0.7805 \\
 & & & & & A--IV & 87\,323 & 0.8325 & 0 & 0.8325 \\
\midrule
\multirow{4}{*}{Lombardia 1} & \multirow{4}{*}{16} & \multirow{4}{*}{1\,874\,918} & \multirow{4}{*}{117\,182} & \multirow{4}{*}{1} & CR & 846\,248 & 7.2217 & 7 & 0.2217 \\
 & & & & & CL & 614\,824 & 5.2467 & 5 & 0.2467 \\
 & & & & & M5S & 177\,990 & 1.5189 & 1 & 0.5189 \\
 & & & & & A--IV & 235\,856 & 2.0127 & 2 & 0.0127 \\
\midrule
\multirow{4}{*}{Lombardia 2} & \multirow{4}{*}{9} & \multirow{4}{*}{1\,009\,841} & \multirow{4}{*}{112\,204} & \multirow{4}{*}{2} & CR & 584\,781 & 5.2118 & 5 & 0.2118 \\
 & & & & & CL & 246\,174 & 2.1940 & 2 & 0.1940 \\
 & & & & & M5S & 71\,227 & 0.6348 & 0 & 0.6348 \\
 & & & & & A--IV & 107\,659 & 0.9595 & 0 & 0.9595 \\
\midrule
\multirow{4}{*}{Lombardia 3} & \multirow{4}{*}{9} & \multirow{4}{*}{1\,098\,730} & \multirow{4}{*}{122\,081} & \multirow{4}{*}{2} & CR & 649\,167 & 5.3175 & 5 & 0.3175 \\
 & & & & & CL & 269\,024 & 2.2037 & 2 & 0.2037 \\
 & & & & & M5S & 67\,539 & 0.5532 & 0 & 0.5532 \\
 & & & & & A--IV & 113\,000 & 0.9256 & 0 & 0.9256 \\
\midrule
\multirow{4}{*}{Lombardia 4} & \multirow{4}{*}{7} & \multirow{4}{*}{764\,979} & \multirow{4}{*}{109\,282} & \multirow{4}{*}{3} & CR & 433\,909 & 3.9705 & 3 & 0.9705 \\
 & & & & & CL & 203\,079 & 1.8583 & 1 & 0.8583 \\
 & & & & & M5S & 62\,129 & 0.5685 & 0 & 0.5685 \\
 & & & & & A--IV & 65\,862 & 0.6027 & 0 & 0.6027 \\
\midrule
\multirow{4}{*}{Veneto 1} & \multirow{4}{*}{8} & \multirow{4}{*}{878\,318} & \multirow{4}{*}{109\,789} & \multirow{4}{*}{2} & CR & 512\,664 & 4.6695 & 4 & 0.6695 \\
 & & & & & CL & 228\,674 & 2.0829 & 2 & 0.0828 \\
 & & & & & M5S & 58\,584 & 0.5336 & 0 & 0.5336 \\
 & & & & & A--IV & 78\,396 & 0.7141 & 0 & 0.7141 \\
\midrule
\multirow{4}{*}{Veneto 2} & \multirow{4}{*}{12} & \multirow{4}{*}{1\,410\,216} & \multirow{4}{*}{117\,518} & \multirow{4}{*}{2} & CR & 849\,166 & 7.2258 & 7 & 0.2258 \\
 & & & & & CL & 340\,991 & 2.9016 & 2 & 0.9016 \\
 & & & & & M5S & 87\,735 & 0.7466 & 0 & 0.7466 \\
 & & & & & A--IV & 132\,324 & 1.1260 & 1 & 0.1260 \\
\midrule
\multirow{4}{*}{Friuli-Venezia Giulia} & \multirow{4}{*}{5} & \multirow{4}{*}{534\,142} & \multirow{4}{*}{106\,828} & \multirow{4}{*}{2} & CR & 289\,799 & 2.7128 & 2 & 0.7128 \\
 & & & & & CL & 150\,144 & 1.4055 & 1 & 0.4055 \\
 & & & & & M5S & 42\,575 & 0.3985 & 0 & 0.3985 \\
 & & & & & A--IV & 51\,624 & 0.4832 & 0 & 0.4832 \\
\midrule
\multirow{4}{*}{Liguria} & \multirow{4}{*}{6} & \multirow{4}{*}{664\,731} & \multirow{4}{*}{110\,788} & \multirow{4}{*}{2} & CR & 293\,569 & 2.6498 & 2 & 0.6498 \\
 & & & & & CL & 223\,395 & 2.0164 & 2 & 0.0164 \\
 & & & & & M5S & 93\,605 & 0.8449 & 0 & 0.8449 \\
 & & & & & A--IV & 54\,162 & 0.4889 & 0 & 0.4889 \\
\midrule
\multirow{4}{*}{Emilia-Romagna} & \multirow{4}{*}{18} & \multirow{4}{*}{2\,132\,108} & \multirow{4}{*}{118\,450} & \multirow{4}{*}{3} & CR & 885\,124 & 7.4726 & 7 & 0.4726 \\
 & & & & & CL & 820\,983 & 6.9311 & 6 & 0.9311 \\
 & & & & & M5S & 228\,663 & 1.9305 & 1 & 0.9305 \\
 & & & & & A--IV & 197\,338 & 1.6660 & 1 & 0.6660 \\
\midrule
\multirow{4}{*}{Toscana} & \multirow{4}{*}{15} & \multirow{4}{*}{1\,743\,757} & \multirow{4}{*}{116\,250} & \multirow{4}{*}{2} & CR & 715\,363 & 6.1537 & 6 & 0.1537 \\
 & & & & & CL & 642\,631 & 5.5280 & 5 & 0.5280 \\
 & & & & & M5S & 209\,240 & 1.7999 & 1 & 0.7999 \\
 & & & & & A--IV & 176\,523 & 1.5185 & 1 & 0.5185 \\
\midrule
\multirow{4}{*}{Umbria} & \multirow{4}{*}{4} & \multirow{4}{*}{404\,364} & \multirow{4}{*}{101\,091} & \multirow{4}{*}{2} & CR & 197\,922 & 1.9579 & 1 & 0.9579 \\
 & & & & & CL & 115\,646 & 1.1440 & 1 & 0.1440 \\
 & & & & & M5S & 55\,195 & 0.5460 & 0 & 0.5460 \\
 & & & & & A--IV & 35\,601 & 0.3522 & 0 & 0.3522 \\
\midrule
\multirow{4}{*}{Marche} & \multirow{4}{*}{6} & \multirow{4}{*}{693\,952} & \multirow{4}{*}{115\,658} & \multirow{4}{*}{3} & CR & 334\,208 & 2.8896 & 2 & 0.8896 \\
 & & & & & CL & 199\,721 & 1.7268 & 1 & 0.7268 \\
 & & & & & M5S & 103\,594 & 0.8957 & 0 & 0.8957 \\
 & & & & & A--IV & 56\,429 & 0.4879 & 0 & 0.4879 \\
\midrule
\multirow{4}{*}{Lazio 1} & \multirow{4}{*}{15} & \multirow{4}{*}{1\,688\,628} & \multirow{4}{*}{112\,575} & \multirow{4}{*}{2} & CR & 723\,572 & 6.4275 & 6 & 0.4275 \\
 & & & & & CL & 529\,154 & 4.7005 & 4 & 0.7005 \\
 & & & & & M5S & 266\,132 & 2.3640 & 2 & 0.3640 \\
 & & & & & A--IV & 169\,770 & 1.5081 & 1 & 0.5081 \\
\midrule
\multirow{4}{*}{Lazio 2} & \multirow{4}{*}{7} & \multirow{4}{*}{853\,494} & \multirow{4}{*}{121\,927} & \multirow{4}{*}{2} & CR & 477\,928 & 3.9198 & 3 & 0.9198 \\
 & & & & & CL & 179\,007 & 1.4681 & 1 & 0.4681 \\
 & & & & & M5S & 140\,163 & 1.1496 & 1 & 0.1496 \\
 & & & & & A--IV & 56\,396 & 0.4625 & 0 & 0.4625 \\
\midrule
\multirow{4}{*}{Abruzzo} & \multirow{4}{*}{6} & \multirow{4}{*}{582\,526} & \multirow{4}{*}{97\,087} & \multirow{4}{*}{1} & CR & 294\,410 & 3.0324 & 3 & 0.0324 \\
 & & & & & CL & 133\,349 & 1.3735 & 1 & 0.3735 \\
 & & & & & M5S & 115\,456 & 1.1892 & 1 & 0.1892 \\
 & & & & & A--IV & 39\,311 & 0.4049 & 0 & 0.4049 \\
\midrule
\multirow{4}{*}{Molise} & \multirow{4}{*}{1} & \multirow{4}{*}{120\,368} & \multirow{4}{*}{120\,368} & \multirow{4}{*}{1} & CR & 53\,375 & 0.4434 & 0 & 0.4434 \\
 & & & & & CL & 29\,302 & 0.2434 & 0 & 0.2434 \\
 & & & & & M5S & 31\,441 & 0.2612 & 0 & 0.2612 \\
 & & & & & A--IV & 6\,250 & 0.0519 & 0 & 0.0519 \\
\midrule
\multirow{4}{*}{Campania 1} & \multirow{4}{*}{13} & \multirow{4}{*}{1\,082\,296} & \multirow{4}{*}{83\,253} & \multirow{4}{*}{3} & CR & 307\,088 & 3.6886 & 3 & 0.6886 \\
 & & & & & CL & 228\,721 & 2.7473 & 2 & 0.7473 \\
 & & & & & M5S & 483\,699 & 5.8100 & 5 & 0.8100 \\
 & & & & & A--IV & 62\,788 & 0.7542 & 0 & 0.7542 \\
\midrule
\multirow{4}{*}{Campania 2} & \multirow{4}{*}{11} & \multirow{4}{*}{1\,031\,957} & \multirow{4}{*}{93\,814} & \multirow{4}{*}{2} & CR & 420\,362 & 4.4808 & 4 & 0.4808 \\
 & & & & & CL & 241\,228 & 2.5713 & 2 & 0.5713 \\
 & & & & & M5S & 312\,979 & 3.3362 & 3 & 0.3362 \\
 & & & & & A--IV & 57\,388 & 0.6117 & 0 & 0.6117 \\
\midrule
\multirow{4}{*}{Puglia} & \multirow{4}{*}{17} & \multirow{4}{*}{1\,653\,114} & \multirow{4}{*}{97\,242} & \multirow{4}{*}{2} & CR & 702\,128 & 7.2204 & 7 & 0.2204 \\
 & & & & & CL & 379\,678 & 3.9045 & 3 & 0.9045 \\
 & & & & & M5S & 487\,348 & 5.0117 & 5 & 0.0117 \\
 & & & & & A--IV & 83\,960 & 0.8634 & 0 & 0.8634 \\
\midrule
\multirow{4}{*}{Basilicata} & \multirow{4}{*}{3} & \multirow{4}{*}{224\,799} & \multirow{4}{*}{74\,933} & \multirow{4}{*}{2} & CR & 89\,216 & 1.1906 & 1 & 0.1906 \\
 & & & & & CL & 50\,599 & 0.6753 & 0 & 0.6753 \\
 & & & & & M5S & 61\,114 & 0.8156 & 0 & 0.8156 \\
 & & & & & A--IV & 23\,870 & 0.3186 & 0 & 0.3186 \\
\midrule
\multirow{4}{*}{Calabria} & \multirow{4}{*}{8} & \multirow{4}{*}{655\,750} & \multirow{4}{*}{81\,968} & \multirow{4}{*}{2} & CR & 290\,365 & 3.5424 & 3 & 0.5424 \\
 & & & & & CL & 124\,185 & 1.5150 & 1 & 0.5150 \\
 & & & & & M5S & 211\,390 & 2.5789 & 2 & 0.5789 \\
 & & & & & A--IV & 29\,810 & 0.3637 & 0 & 0.3637 \\
\midrule
\multirow{4}{*}{Sicilia 1} & \multirow{4}{*}{9} & \multirow{4}{*}{814\,168} & \multirow{4}{*}{90\,463} & \multirow{4}{*}{2} & CR & 320\,661 & 3.5447 & 3 & 0.5447 \\
 & & & & & CL & 151\,713 & 1.6771 & 1 & 0.6771 \\
 & & & & & M5S & 288\,531 & 3.1895 & 3 & 0.1895 \\
 & & & & & A--IV & 53\,263 & 0.5888 & 0 & 0.5888 \\
\midrule
\multirow{4}{*}{Sicilia 2} & \multirow{4}{*}{11} & \multirow{4}{*}{908\,228} & \multirow{4}{*}{82\,566} & \multirow{4}{*}{2} & CR & 401\,381 & 4.8613 & 4 & 0.8613 \\
 & & & & & CL & 168\,282 & 2.0382 & 2 & 0.0382 \\
 & & & & & M5S & 287\,274 & 3.4793 & 3 & 0.4793 \\
 & & & & & A--IV & 51\,291 & 0.6212 & 0 & 0.6212 \\
\midrule
\multirow{4}{*}{Sardegna} & \multirow{4}{*}{7} & \multirow{4}{*}{623\,528} & \multirow{4}{*}{89\,075} & \multirow{4}{*}{2} & CR & 263\,756 & 2.9611 & 2 & 0.9611 \\
 & & & & & CL & 178\,702 & 2.0062 & 2 & 0.0062 \\
 & & & & & M5S & 149\,477 & 1.6781 & 1 & 0.6781 \\
 & & & & & A--IV & 31\,593 & 0.3547 & 0 & 0.3547 \\
\midrule
\multirow{5}{*}{Trentino-A.A.} & \multirow{5}{*}{3} & \multirow{5}{*}{458\,947} & \multirow{5}{*}{152\,982} & \multirow{5}{*}{2} & CR & 154\,877 & 1.0124 & 1 & 0.0124 \\
 & & & & & CL & 130\,988 & 0.8562 & 0 & 0.8562 \\
 & & & & & M5S & 25\,394 & 0.1660 & 0 & 0.1660 \\
 & & & & & A--IV & 30\,678 & 0.2005 & 0 & 0.2005 \\
 & & & & & SVP & 117\,010 & 0.7649 & 0 & 0.7649 \\
\end{longtable}
\end{small}

\section{Final unit-level assignments of the four runs}
\label{app:assh}

Table~\ref{tab:assh} reports, for each constituency and unit, the final seats after residual assignment \emph{and} compensation (Table~\ref{tab:comp}) under the four runs: C and A with the official order, A with the reversed order, B with the official order. Column sums per constituency equal $s_c$ except under B in Sicilia 2 and Sardegna (one seat lost in each); row sums per unit equal the national entitlements 114/68/41/21/1 except CR $=112$ under B.

% tabella generata automaticamente da genera_tabelle.py
\begin{small}
\begin{longtable}{@{}llr@{\hspace{1.4em}}rrrr@{}}
\caption{Final seat assignments of Art.~83(1)(h) (after residual seats and statutory
compensation) for each constituency and unit, under the four runs: C (official order),
A (official order), A (reversed order), B (official order). ``Int.''\ is the
order-independent integer-part assignment of Table~\ref{tab:letth-full}. Rows marked
$^{*}$ differ across runs.}
\label{tab:assh}\\
\toprule
Constituency & Unit & Int. & C & A & A$_{\mathrm{rev}}$ & B \\
\midrule
\endfirsthead
\multicolumn{7}{l}{\small\emph{Table \ref{tab:assh} (continued)}}\\
\toprule
Constituency & Unit & Int. & C & A & A$_{\mathrm{rev}}$ & B \\
\midrule
\endhead
\bottomrule
\endlastfoot
\multirow{4}{*}{Piemonte 1} & CR$^{*}$ & 4 & 4 & 4 & 5 & 4 \\
 & CL$^{*}$ & 3 & 4 & 4 & 3 & 4 \\
 & M5S & 1 & 1 & 1 & 1 & 1 \\
 & A--IV & 0 & 1 & 1 & 1 & 1 \\
\midrule
\multirow{4}{*}{Piemonte 2} & CR & 5 & 5 & 5 & 5 & 5 \\
 & CL$^{*}$ & 2 & 2 & 2 & 3 & 2 \\
 & M5S$^{*}$ & 0 & 1 & 1 & 0 & 1 \\
 & A--IV & 0 & 1 & 1 & 1 & 1 \\
\midrule
\multirow{4}{*}{Lombardia 1} & CR$^{*}$ & 7 & 8 & 8 & 8 & 7 \\
 & CL & 5 & 5 & 5 & 5 & 5 \\
 & M5S$^{*}$ & 1 & 1 & 1 & 1 & 2 \\
 & A--IV & 2 & 2 & 2 & 2 & 2 \\
\midrule
\multirow{4}{*}{Lombardia 2} & CR$^{*}$ & 5 & 6 & 5 & 5 & 5 \\
 & CL & 2 & 2 & 2 & 2 & 2 \\
 & M5S$^{*}$ & 0 & 0 & 1 & 1 & 1 \\
 & A--IV & 0 & 1 & 1 & 1 & 1 \\
\midrule
\multirow{4}{*}{Lombardia 3} & CR$^{*}$ & 5 & 6 & 6 & 6 & 5 \\
 & CL & 2 & 2 & 2 & 2 & 2 \\
 & M5S$^{*}$ & 0 & 0 & 0 & 0 & 1 \\
 & A--IV & 0 & 1 & 1 & 1 & 1 \\
\midrule
\multirow{4}{*}{Lombardia 4} & CR & 3 & 4 & 4 & 4 & 4 \\
 & CL & 1 & 2 & 2 & 2 & 2 \\
 & M5S & 0 & 0 & 0 & 0 & 0 \\
 & A--IV & 0 & 1 & 1 & 1 & 1 \\
\midrule
\multirow{4}{*}{Veneto 1} & CR & 4 & 5 & 5 & 5 & 5 \\
 & CL & 2 & 2 & 2 & 2 & 2 \\
 & M5S & 0 & 0 & 0 & 0 & 0 \\
 & A--IV & 0 & 1 & 1 & 1 & 1 \\
\midrule
\multirow{4}{*}{Veneto 2} & CR & 7 & 7 & 7 & 7 & 7 \\
 & CL & 2 & 3 & 3 & 3 & 3 \\
 & M5S & 0 & 1 & 1 & 1 & 1 \\
 & A--IV & 1 & 1 & 1 & 1 & 1 \\
\midrule
\multirow{4}{*}{Friuli-Venezia Giulia} & CR & 2 & 3 & 3 & 3 & 3 \\
 & CL & 1 & 1 & 1 & 1 & 1 \\
 & M5S & 0 & 0 & 0 & 0 & 0 \\
 & A--IV & 0 & 1 & 1 & 1 & 1 \\
\midrule
\multirow{4}{*}{Liguria} & CR & 2 & 3 & 3 & 3 & 3 \\
 & CL & 2 & 2 & 2 & 2 & 2 \\
 & M5S & 0 & 1 & 1 & 1 & 1 \\
 & A--IV & 0 & 0 & 0 & 0 & 0 \\
\midrule
\multirow{4}{*}{Emilia-Romagna} & CR & 7 & 7 & 7 & 7 & 7 \\
 & CL & 6 & 7 & 7 & 7 & 7 \\
 & M5S & 1 & 2 & 2 & 2 & 2 \\
 & A--IV & 1 & 2 & 2 & 2 & 2 \\
\midrule
\multirow{4}{*}{Toscana} & CR & 6 & 6 & 6 & 6 & 6 \\
 & CL & 5 & 6 & 6 & 6 & 6 \\
 & M5S & 1 & 2 & 2 & 2 & 2 \\
 & A--IV & 1 & 1 & 1 & 1 & 1 \\
\midrule
\multirow{4}{*}{Umbria} & CR & 1 & 2 & 2 & 2 & 2 \\
 & CL & 1 & 1 & 1 & 1 & 1 \\
 & M5S & 0 & 1 & 1 & 1 & 1 \\
 & A--IV & 0 & 0 & 0 & 0 & 0 \\
\midrule
\multirow{4}{*}{Marche} & CR & 2 & 3 & 3 & 3 & 3 \\
 & CL & 1 & 2 & 2 & 2 & 2 \\
 & M5S & 0 & 1 & 1 & 1 & 1 \\
 & A--IV & 0 & 0 & 0 & 0 & 0 \\
\midrule
\multirow{4}{*}{Lazio 1} & CR & 6 & 6 & 6 & 6 & 6 \\
 & CL & 4 & 5 & 5 & 5 & 5 \\
 & M5S & 2 & 2 & 2 & 2 & 2 \\
 & A--IV & 1 & 2 & 2 & 2 & 2 \\
\midrule
\multirow{4}{*}{Lazio 2} & CR & 3 & 4 & 4 & 4 & 4 \\
 & CL & 1 & 2 & 2 & 2 & 2 \\
 & M5S & 1 & 1 & 1 & 1 & 1 \\
 & A--IV & 0 & 0 & 0 & 0 & 0 \\
\midrule
\multirow{4}{*}{Abruzzo} & CR & 3 & 3 & 3 & 3 & 3 \\
 & CL & 1 & 1 & 1 & 1 & 1 \\
 & M5S & 1 & 1 & 1 & 1 & 1 \\
 & A--IV & 0 & 1 & 1 & 1 & 1 \\
\midrule
\multirow{4}{*}{Molise} & CR & 0 & 1 & 1 & 1 & 1 \\
 & CL & 0 & 0 & 0 & 0 & 0 \\
 & M5S & 0 & 0 & 0 & 0 & 0 \\
 & A--IV & 0 & 0 & 0 & 0 & 0 \\
\midrule
\multirow{4}{*}{Campania 1} & CR & 3 & 3 & 3 & 3 & 3 \\
 & CL & 2 & 3 & 3 & 3 & 3 \\
 & M5S & 5 & 6 & 6 & 6 & 6 \\
 & A--IV & 0 & 1 & 1 & 1 & 1 \\
\midrule
\multirow{4}{*}{Campania 2} & CR & 4 & 4 & 4 & 4 & 4 \\
 & CL & 2 & 3 & 3 & 3 & 3 \\
 & M5S & 3 & 3 & 3 & 3 & 3 \\
 & A--IV & 0 & 1 & 1 & 1 & 1 \\
\midrule
\multirow{4}{*}{Puglia} & CR & 7 & 7 & 7 & 7 & 7 \\
 & CL & 3 & 4 & 4 & 4 & 4 \\
 & M5S & 5 & 5 & 5 & 5 & 5 \\
 & A--IV & 0 & 1 & 1 & 1 & 1 \\
\midrule
\multirow{4}{*}{Basilicata} & CR & 1 & 1 & 1 & 1 & 1 \\
 & CL & 0 & 1 & 1 & 1 & 1 \\
 & M5S$^{*}$ & 0 & 1 & 1 & 1 & 0 \\
 & A--IV$^{*}$ & 0 & 0 & 0 & 0 & 1 \\
\midrule
\multirow{4}{*}{Calabria} & CR & 3 & 4 & 4 & 4 & 4 \\
 & CL$^{*}$ & 1 & 1 & 1 & 1 & 2 \\
 & M5S$^{*}$ & 2 & 3 & 3 & 3 & 2 \\
 & A--IV & 0 & 0 & 0 & 0 & 0 \\
\midrule
\multirow{4}{*}{Sicilia 1} & CR & 3 & 3 & 3 & 3 & 3 \\
 & CL & 1 & 2 & 2 & 2 & 2 \\
 & M5S & 3 & 3 & 3 & 3 & 3 \\
 & A--IV & 0 & 1 & 1 & 1 & 1 \\
\midrule
\multirow{4}{*}{Sicilia 2} & CR & 4 & 5 & 5 & 5 & 5 \\
 & CL & 2 & 2 & 2 & 2 & 2 \\
 & M5S & 3 & 3 & 3 & 3 & 3 \\
 & A--IV$^{*}$ & 0 & 1 & 1 & 1 & 0 \\
\midrule
\multirow{4}{*}{Sardegna} & CR & 2 & 3 & 3 & 3 & 3 \\
 & CL$^{*}$ & 2 & 2 & 3 & 2 & 2 \\
 & M5S$^{*}$ & 1 & 2 & 1 & 2 & 1 \\
 & A--IV & 0 & 0 & 0 & 0 & 0 \\
\midrule
\multirow{5}{*}{Trentino-A.A.} & CR$^{*}$ & 1 & 1 & 2 & 1 & 2 \\
 & CL$^{*}$ & 0 & 1 & 0 & 1 & 0 \\
 & M5S & 0 & 0 & 0 & 0 & 0 \\
 & A--IV & 0 & 0 & 0 & 0 & 0 \\
 & SVP & 0 & 1 & 1 & 1 & 1 \\
\end{longtable}
\end{small}

\section{The within-coalition distribution (Art.~83(1)(i))}
\label{app:i}

Table~\ref{tab:icifre} publishes the constituency tallies of the five above-threshold coalition lists, the inputs of \arti{}. For any run, the reader can re-execute the within-coalition distribution of a coalition in a constituency by taking the coalition's seat count $s$ from Table~\ref{tab:assh}, computing the quotient $\lfloor(\text{sum of its lists' tallies there})/s\rfloor$, assigning integer parts and residual seats per the interpretation under test, and applying the compensation transfers of Table~\ref{tab:comp}. Table~\ref{tab:assi} gives the resulting final per-list assignments of all four runs for verification.

% tabella generata automaticamente da genera_tabelle.py
\begin{small}
\begin{longtable}{@{}lrrr@{\hspace{1.6em}}rr@{}}
\caption{Constituency tallies (\emph{cifre elettorali circoscrizionali}) of the
above-threshold lists of the two coalitions, used by Art.~83(1)(i). Together with the
per-run unit seats of Table~\ref{tab:assh}, these are the complete inputs of the
within-coalition distribution: for a coalition with $s$ seats in a constituency, the
quotient is the truncated ratio of the sum of its listed tallies to $s$.}
\label{tab:icifre}\\
\toprule
 & \multicolumn{3}{c}{Centre-right} & \multicolumn{2}{c}{Centre-left} \\
\cmidrule(lr){2-4}\cmidrule(l){5-6}
Constituency & FdI & Lega & FI & PD & AVS \\
\midrule
\endfirsthead
\multicolumn{6}{l}{\small\emph{Table \ref{tab:icifre} (continued)}}\\
\toprule
Constituency & FdI & Lega & FI & PD & AVS \\
\midrule
\endhead
\bottomrule
\endlastfoot
Piemonte 1 & 260\,708 & 93\,326 & 74\,640 & 240\,909 & 53\,840 \\
Piemonte 2 & 306\,918 & 131\,067 & 91\,128 & 178\,530 & 31\,331 \\
Lombardia 1 & 498\,152 & 199\,816 & 148\,280 & 436\,481 & 93\,213 \\
Lombardia 2 & 326\,768 & 168\,427 & 89\,586 & 172\,356 & 36\,917 \\
Lombardia 3 & 366\,787 & 190\,804 & 91\,576 & 195\,963 & 39\,145 \\
Lombardia 4 & 251\,984 & 112\,767 & 69\,158 & 157\,094 & 23\,664 \\
Veneto 1 & 311\,663 & 140\,204 & 60\,797 & 163\,298 & 34\,530 \\
Veneto 2 & 509\,920 & 224\,986 & 114\,260 & 245\,703 & 48\,896 \\
Friuli-Venezia Giulia & 185\,394 & 64\,806 & 39\,599 & 108\,870 & 21\,986 \\
Liguria & 178\,239 & 68\,186 & 47\,144 & 166\,543 & 32\,131 \\
Emilia-Romagna & 577\,234 & 173\,508 & 134\,382 & 648\,373 & 99\,744 \\
Toscana & 487\,339 & 123\,321 & 104\,703 & 495\,637 & 92\,592 \\
Umbria & 134\,357 & 33\,776 & 29\,789 & 91\,052 & 15\,422 \\
Marche & 222\,249 & 60\,379 & 51\,580 & 155\,263 & 25\,348 \\
Lazio 1 & 539\,925 & 89\,832 & 93\,815 & 388\,019 & 79\,145 \\
Lazio 2 & 305\,414 & 80\,663 & 91\,851 & 135\,236 & 25\,459 \\
Abruzzo & 174\,548 & 50\,374 & 69\,488 & 104\,047 & 16\,946 \\
Molise & 27\,632 & 11\,042 & 14\,701 & 23\,444 & 3\,815 \\
Campania 1 & 161\,849 & 33\,432 & 111\,807 & 168\,654 & 35\,867 \\
Campania 2 & 239\,915 & 68\,402 & 112\,045 & 191\,661 & 27\,549 \\
Puglia & 410\,718 & 91\,808 & 199\,602 & 293\,310 & 52\,881 \\
Basilicata & 44\,396 & 21\,946 & 22\,874 & 37\,147 & 8\,348 \\
Calabria & 136\,479 & 41\,673 & 112\,213 & 103\,315 & 13\,050 \\
Sicilia 1 & 167\,857 & 44\,265 & 108\,539 & 114\,682 & 19\,768 \\
Sicilia 2 & 221\,779 & 59\,451 & 120\,151 & 127\,820 & 22\,483 \\
Sardegna & 161\,971 & 43\,117 & 58\,668 & 128\,596 & 34\,675 \\
Trentino-A.A. & 94\,823 & 43\,109 & 16\,945 & 86\,459 & 29\,599 \\
\end{longtable}
\end{small}

% tabella generata automaticamente da genera_tabelle.py
\begin{small}
\begin{longtable}{@{}ll@{\hspace{1.4em}}rrrr@{}}
\caption{Final within-coalition seat assignments of Art.~83(1)(i) (after residual seats
and statutory compensation) for the five coalition lists, under the four runs. Rows
marked $^{*}$ differ across runs. Single lists (M5S, Azione--Italia Viva, SVP--PATT) take
their unit seats of Table~\ref{tab:assh} directly and are omitted.}
\label{tab:assi}\\
\toprule
Constituency & List & C & A & A$_{\mathrm{rev}}$ & B \\
\midrule
\endfirsthead
\multicolumn{6}{l}{\small\emph{Table \ref{tab:assi} (continued)}}\\
\toprule
Constituency & List & C & A & A$_{\mathrm{rev}}$ & B \\
\midrule
\endhead
\bottomrule
\endlastfoot
\multirow{5}{*}{Piemonte 1} & FdI$^{*}$ & 2 & 2 & 3 & 2 \\
 & Lega & 1 & 1 & 1 & 1 \\
 & FI & 1 & 1 & 1 & 1 \\
 & PD$^{*}$ & 3 & 3 & 2 & 3 \\
 & AVS & 1 & 1 & 1 & 1 \\
\midrule
\multirow{5}{*}{Piemonte 2} & FdI & 3 & 3 & 3 & 3 \\
 & Lega & 1 & 1 & 1 & 1 \\
 & FI & 1 & 1 & 1 & 1 \\
 & PD & 2 & 2 & 2 & 2 \\
 & AVS$^{*}$ & 0 & 0 & 1 & 0 \\
\midrule
\multirow{5}{*}{Lombardia 1} & FdI$^{*}$ & 5 & 5 & 5 & 4 \\
 & Lega & 2 & 2 & 2 & 2 \\
 & FI & 1 & 1 & 1 & 1 \\
 & PD & 4 & 4 & 4 & 4 \\
 & AVS & 1 & 1 & 1 & 1 \\
\midrule
\multirow{5}{*}{Lombardia 2} & FdI & 3 & 3 & 3 & 3 \\
 & Lega$^{*}$ & 2 & 1 & 1 & 1 \\
 & FI & 1 & 1 & 1 & 1 \\
 & PD$^{*}$ & 1 & 2 & 2 & 2 \\
 & AVS$^{*}$ & 1 & 0 & 0 & 0 \\
\midrule
\multirow{4}{*}{Lombardia 3} & FdI & 3 & 3 & 3 & 3 \\
 & Lega$^{*}$ & 2 & 2 & 2 & 1 \\
 & FI & 1 & 1 & 1 & 1 \\
 & PD & 2 & 2 & 2 & 2 \\
\midrule
\multirow{4}{*}{Lombardia 4} & FdI & 2 & 2 & 2 & 2 \\
 & Lega & 1 & 1 & 1 & 1 \\
 & FI & 1 & 1 & 1 & 1 \\
 & PD & 2 & 2 & 2 & 2 \\
\midrule
\multirow{4}{*}{Veneto 1} & FdI & 3 & 3 & 3 & 3 \\
 & Lega & 1 & 1 & 1 & 1 \\
 & FI & 1 & 1 & 1 & 1 \\
 & PD & 2 & 2 & 2 & 2 \\
\midrule
\multirow{5}{*}{Veneto 2} & FdI & 4 & 4 & 4 & 4 \\
 & Lega & 2 & 2 & 2 & 2 \\
 & FI & 1 & 1 & 1 & 1 \\
 & PD$^{*}$ & 2 & 2 & 2 & 3 \\
 & AVS$^{*}$ & 1 & 1 & 1 & 0 \\
\midrule
\multirow{3}{*}{Friuli-Venezia Giulia} & FdI & 2 & 2 & 2 & 2 \\
 & Lega & 1 & 1 & 1 & 1 \\
 & PD & 1 & 1 & 1 & 1 \\
\midrule
\multirow{3}{*}{Liguria} & FdI & 2 & 2 & 2 & 2 \\
 & Lega & 1 & 1 & 1 & 1 \\
 & PD & 2 & 2 & 2 & 2 \\
\midrule
\multirow{5}{*}{Emilia-Romagna} & FdI & 5 & 5 & 5 & 5 \\
 & Lega & 1 & 1 & 1 & 1 \\
 & FI & 1 & 1 & 1 & 1 \\
 & PD & 6 & 6 & 6 & 6 \\
 & AVS & 1 & 1 & 1 & 1 \\
\midrule
\multirow{5}{*}{Toscana} & FdI & 4 & 4 & 4 & 4 \\
 & Lega & 1 & 1 & 1 & 1 \\
 & FI & 1 & 1 & 1 & 1 \\
 & PD & 5 & 5 & 5 & 5 \\
 & AVS & 1 & 1 & 1 & 1 \\
\midrule
\multirow{4}{*}{Umbria} & FdI$^{*}$ & 1 & 1 & 1 & 2 \\
 & Lega$^{*}$ & 0 & 0 & 1 & 0 \\
 & FI$^{*}$ & 1 & 1 & 0 & 0 \\
 & PD & 1 & 1 & 1 & 1 \\
\midrule
\multirow{3}{*}{Marche} & FdI & 2 & 2 & 2 & 2 \\
 & Lega & 1 & 1 & 1 & 1 \\
 & PD & 2 & 2 & 2 & 2 \\
\midrule
\multirow{5}{*}{Lazio 1} & FdI & 4 & 4 & 4 & 4 \\
 & Lega & 1 & 1 & 1 & 1 \\
 & FI & 1 & 1 & 1 & 1 \\
 & PD & 4 & 4 & 4 & 4 \\
 & AVS & 1 & 1 & 1 & 1 \\
\midrule
\multirow{4}{*}{Lazio 2} & FdI & 2 & 2 & 2 & 2 \\
 & Lega & 1 & 1 & 1 & 1 \\
 & FI & 1 & 1 & 1 & 1 \\
 & PD & 2 & 2 & 2 & 2 \\
\midrule
\multirow{3}{*}{Abruzzo} & FdI & 2 & 2 & 2 & 2 \\
 & FI & 1 & 1 & 1 & 1 \\
 & PD & 1 & 1 & 1 & 1 \\
\midrule
\multirow{1}{*}{Molise} & FdI & 1 & 1 & 1 & 1 \\
\midrule
\multirow{4}{*}{Campania 1} & FdI & 2 & 2 & 2 & 2 \\
 & FI & 1 & 1 & 1 & 1 \\
 & PD & 2 & 2 & 2 & 2 \\
 & AVS & 1 & 1 & 1 & 1 \\
\midrule
\multirow{5}{*}{Campania 2} & FdI & 2 & 2 & 2 & 2 \\
 & Lega & 1 & 1 & 1 & 1 \\
 & FI & 1 & 1 & 1 & 1 \\
 & PD$^{*}$ & 2 & 2 & 2 & 3 \\
 & AVS$^{*}$ & 1 & 1 & 1 & 0 \\
\midrule
\multirow{5}{*}{Puglia} & FdI & 4 & 4 & 4 & 4 \\
 & Lega & 1 & 1 & 1 & 1 \\
 & FI & 2 & 2 & 2 & 2 \\
 & PD & 3 & 3 & 3 & 3 \\
 & AVS & 1 & 1 & 1 & 1 \\
\midrule
\multirow{3}{*}{Basilicata} & FdI$^{*}$ & 1 & 1 & 0 & 1 \\
 & FI$^{*}$ & 0 & 0 & 1 & 0 \\
 & PD & 1 & 1 & 1 & 1 \\
\midrule
\multirow{5}{*}{Calabria} & FdI & 2 & 2 & 2 & 2 \\
 & Lega & 1 & 1 & 1 & 1 \\
 & FI & 1 & 1 & 1 & 1 \\
 & PD & 1 & 1 & 1 & 1 \\
 & AVS$^{*}$ & 0 & 0 & 0 & 1 \\
\midrule
\multirow{4}{*}{Sicilia 1} & FdI & 2 & 2 & 2 & 2 \\
 & FI & 1 & 1 & 1 & 1 \\
 & PD$^{*}$ & 2 & 2 & 2 & 1 \\
 & AVS$^{*}$ & 0 & 0 & 0 & 1 \\
\midrule
\multirow{5}{*}{Sicilia 2} & FdI & 3 & 3 & 3 & 3 \\
 & Lega & 1 & 1 & 1 & 1 \\
 & FI & 1 & 1 & 1 & 1 \\
 & PD$^{*}$ & 2 & 1 & 2 & 1 \\
 & AVS$^{*}$ & 0 & 1 & 0 & 1 \\
\midrule
\multirow{4}{*}{Sardegna} & FdI & 2 & 2 & 2 & 2 \\
 & FI & 1 & 1 & 1 & 1 \\
 & PD$^{*}$ & 1 & 2 & 1 & 1 \\
 & AVS & 1 & 1 & 1 & 1 \\
\midrule
\multirow{3}{*}{Trentino-A.A.} & FdI & 1 & 1 & 1 & 1 \\
 & Lega$^{*}$ & 0 & 1 & 0 & 1 \\
 & PD$^{*}$ & 1 & 0 & 1 & 0 \\
\end{longtable}
\end{small}

\section{Per-deputy election frequencies under random orders}
\label{app:perm}

Table~\ref{tab:perm-deputati} reports the per-person outcome of the order experiment of Section~\ref{sec:results}: for each of the 119 persons elected by Algorithm~A under some but not all of the 1{,}000 uniformly random constituency orders, the number of orders under which they are elected, together with their membership in the practice parliament (interpretation~C) and in the run of A under the reference order. The other 339 modelled deputies are elected under every sampled order.

% tabella generata automaticamente da genera_tab_perm.py
\begin{small}
\begin{longtable}{@{}llrrcc@{}}
\caption{Per-deputy election frequencies under Algorithm A over the
1{,}000 uniformly random constituency orders of Section~\ref{sec:results}:
the 119 persons elected under some sampled orders but not others, by
decreasing frequency. The remaining 339 modelled deputies (including all
146 single-member winners) are elected under every sampled order and are
not listed. ``C'': member of the practice parliament (interpretation C,
which is order-independent); ``A\textsubscript{ref}'': elected by
Algorithm A under the reference order. Frequencies are counts over the
sampled orders, not probabilities (see the sampling caveats in
Section~\ref{sec:results}); the archived per-person table also lists the
colleges of election with their frequencies (Section~\ref{sec:repro}).}
\label{tab:perm-deputati}\\
\toprule
Deputy & List & Constituency & $n/1000$ & C & A\textsubscript{ref} \\
\midrule
\endfirsthead
\multicolumn{6}{@{}l}{\small\emph{Table~\ref{tab:perm-deputati}, continued.}}\\
\toprule
Deputy & List & Constituency & $n/1000$ & C & A\textsubscript{ref} \\
\midrule
\endhead
\bottomrule
\endlastfoot
Alessia Ambrosi & FdI & Trentino-Alto Adige & 999 & $\bullet$ & $\bullet$ \\
Antonino Iaria & M5S & Piemonte 1 & 999 & $\bullet$ & $\bullet$ \\
Arnaldo Lomuti & M5S & Basilicata & 999 & $\bullet$ & $\bullet$ \\
Gilberto Pichetto Fratin & FI & Piemonte 1 & 998 & $\bullet$ & $\bullet$ \\
Enrico Cappelletti & M5S & Veneto 2 & 997 & $\bullet$ & $\bullet$ \\
Luana Zanella & AVS & Veneto 2 & 996 & $\bullet$ & $\bullet$ \\
Mauro Antonio Donato Laus & PD & Piemonte 1 & 996 & $\bullet$ & $\bullet$ \\
Andrea Rossi & PD & Emilia-Romagna & 994 & $\bullet$ & $\bullet$ \\
Marco Simiani & PD & Toscana & 994 & $\bullet$ & $\bullet$ \\
Matteo Orfini & PD & Lazio 2 & 991 & $\bullet$ & $\bullet$ \\
Lorenzo Guerini & PD & Lombardia 4 & 987 & $\bullet$ & $\bullet$ \\
Roberto Giachetti & A--IV & Lazio 1 & 986 & $\bullet$ & $\bullet$ \\
Susanna Cherchi & M5S & Sardegna & 986 & $\bullet$ & $\bullet$ \\
Irene Manzi & PD & Marche & 984 & $\bullet$ & $\bullet$ \\
Naike Gruppioni & A--IV & Emilia-Romagna & 983 & $\bullet$ & $\bullet$ \\
Luigi Marattin & A--IV & Piemonte 2 & 982 & $\bullet$ & $\bullet$ \\
Daniela Ruffino & A--IV & Piemonte 1 & 980 & $\bullet$ & $\bullet$ \\
Fabrizio Benzoni & A--IV & Lombardia 3 & 978 & $\bullet$ & $\bullet$ \\
Valentina Grippo & A--IV & Veneto 1 & 977 & $\bullet$ & $\bullet$ \\
Isabella De Monte & A--IV & Friuli-Venezia Giulia & 975 & $\bullet$ & $\bullet$ \\
Elisabetta Christiana Lancellotta & FdI & Molise & 963 & $\bullet$ & $\bullet$ \\
Maria Chiara Gadda & A--IV & Lombardia 2 & 961 & $\bullet$ & $\bullet$ \\
Nicola Stumpo & PD & Calabria & 958 & $\bullet$ & $\bullet$ \\
Maurizio Leo & FdI & Sicilia 2 & 956 & $\bullet$ & $\bullet$ \\
Federico Fornaro & PD & Piemonte 2 & 953 & $\bullet$ & $\bullet$ \\
Giulio Cesare Sottanelli & A--IV & Abruzzo & 941 & $\bullet$ & $\bullet$ \\
Vinicio Giuseppe Guido Peluffo & PD & Lombardia 3 & 936 & $\bullet$ & $\bullet$ \\
Vincenzo Amendola & PD & Basilicata & 933 & $\bullet$ & $\bullet$ \\
Andrea Casu & PD & Lazio 1 & 931 & $\bullet$ & $\bullet$ \\
Marco Furfaro & PD & Toscana & 927 & $\bullet$ & $\bullet$ \\
Valentina Ghio & PD & Liguria & 925 & $\bullet$ & $\bullet$ \\
Aldo Mattia & FdI & Basilicata & 922 & $\bullet$ & $\bullet$ \\
Giovanna Iacono & PD & Sicilia 1 & 922 & $\bullet$ & $\bullet$ \\
Piero Franco Rodolfo Fassino & PD & Veneto 1 & 922 & $\bullet$ & $\bullet$ \\
Luciano D'Alfonso & PD & Abruzzo & 917 & $\bullet$ & $\bullet$ \\
Maria Stefania Marino & PD & Sicilia 2 & 917 & $\bullet$ &  \\
Augusto Curti & PD & Marche & 916 & $\bullet$ & $\bullet$ \\
Grazia Di Maggio & FdI & Lombardia 1 & 914 & $\bullet$ & $\bullet$ \\
Emma Pavanelli & M5S & Umbria & 911 & $\bullet$ & $\bullet$ \\
Andrea Quartini & M5S & Toscana & 909 & $\bullet$ & $\bullet$ \\
Sara Ferrari & PD & Trentino-Alto Adige & 909 & $\bullet$ &  \\
Antonella Forattini & PD & Lombardia 4 & 907 & $\bullet$ & $\bullet$ \\
Francesca Ghirra & AVS & Sardegna & 873 & $\bullet$ & $\bullet$ \\
Maria Cecilia Guerra & PD & Piemonte 1 & 873 & $\bullet$ & $\bullet$ \\
Elisa Scutellà & M5S & Calabria & 833 & $\bullet$ & $\bullet$ \\
Valentina Barzotti & M5S & Lombardia 4 & 758 & $\bullet$ & $\bullet$ \\
Roberto Rampi & PD & Lombardia 2 & 707 &  & $\bullet$ \\
Giulio Centemero & Lega & Lombardia 3 & 705 & $\bullet$ & $\bullet$ \\
Francesco Mari & AVS & Campania 2 & 618 & $\bullet$ & $\bullet$ \\
Paolo Pulciani & FdI & Lazio 2 & 585 & $\bullet$ & $\bullet$ \\
Marco Squarta & FdI & Umbria & 399 &  &  \\
Umberto Bossi & Lega & Lombardia 2 & 386 & $\bullet$ &  \\
Rosa D'Amelio & PD & Campania 2 & 382 &  &  \\
Devis Dori & AVS & Lombardia 2 & 307 & $\bullet$ &  \\
Valeria Alessandrini & Lega & Umbria & 301 &  &  \\
Catia Polidori & FI & Umbria & 300 & $\bullet$ & $\bullet$ \\
Samuel Sorial & M5S & Lombardia 3 & 295 &  &  \\
Elisabetta Piccolotti & AVS & Puglia & 167 &  &  \\
Romina Mura & PD & Sardegna & 137 &  & $\bullet$ \\
Michele Casino & FI & Basilicata & 119 &  &  \\
Gabriele Lanzi & M5S & Emilia-Romagna & 99 &  &  \\
Gianni Tonelli & Lega & Emilia-Romagna & 99 &  &  \\
Stella Sorgente & M5S & Toscana & 97 &  &  \\
Giuseppe Buondonno & AVS & Lombardia 4 & 93 &  &  \\
Guido Milanese & FI & Campania 1 & 88 &  &  \\
Donatella Albini & AVS & Lombardia 3 & 86 &  &  \\
Paolo Nicolò Romano & AVS & Piemonte 2 & 86 &  &  \\
Stefano Quaranta & AVS & Liguria & 85 &  &  \\
Luisa Serroni & AVS & Marche & 84 &  &  \\
Pierpaolo Placido Salvatore Montalto & AVS & Sicilia 2 & 83 &  & $\bullet$ \\
Rahel Seium & AVS & Abruzzo & 83 &  &  \\
Serena Pellegrino & AVS & Friuli-Venezia Giulia & 83 &  &  \\
Filippo Sestito & AVS & Calabria & 78 &  &  \\
Marilena Grassadonia & AVS & Sicilia 1 & 78 &  &  \\
Monica Ferraris & M5S & Lombardia 3 & 78 &  &  \\
Giovanni Ravalli & FdI & Piemonte 1 & 77 &  &  \\
Sara Montrasio & M5S & Lombardia 1 & 69 &  &  \\
Luigi D'Eramo & Lega & Abruzzo & 64 &  &  \\
Federica Di Sarcina & AVS & Lazio 2 & 63 &  &  \\
Roberta Rigamonti & AVS & Trentino-Alto Adige & 63 &  &  \\
Enrico Montani & Lega & Piemonte 2 & 61 &  &  \\
Donato Lettieri & AVS & Basilicata & 59 &  &  \\
Lucia Annibali & A--IV & Toscana & 51 &  &  \\
Cosimo Maria Ferri & A--IV & Liguria & 49 &  &  \\
Rosanna Natoli & FdI & Sicilia 2 & 44 &  &  \\
Urania Giulia Rosina Papatheu & FI & Sicilia 2 & 44 &  &  \\
Roberto Nicola Cassinelli & FI & Liguria & 42 &  &  \\
Maria Valentina Vezzali & FI & Marche & 41 &  &  \\
Tommaso Fagioli & A--IV & Marche & 39 &  &  \\
Maria Cristina Sandrin & FdI & Veneto 2 & 38 &  &  \\
Mattia Ierardi & FdI & Veneto 2 & 37 &  &  \\
Giacomo Leonello Leonelli & A--IV & Umbria & 35 &  &  \\
Lara Fadini & Lega & Veneto 2 & 35 &  &  \\
Mario Lolini & Lega & Toscana & 35 &  &  \\
Matteo Bagnoli & FdI & Toscana & 35 &  &  \\
Gabriele Toccafondi & A--IV & Veneto 2 & 34 &  &  \\
Michelina Lunesu & Lega & Sardegna & 34 &  &  \\
Diego Binelli & Lega & Trentino-Alto Adige & 33 &  & $\bullet$ \\
Fulvia Michela Caligiuri & FI & Calabria & 33 &  &  \\
Margherita La Rocca & FI & Sicilia 1 & 33 &  &  \\
Matteo Perego Di Cremnago & FI & Lombardia 1 & 33 &  &  \\
Mario Polese & A--IV & Basilicata & 29 &  &  \\
Annaelsa Tartaglione & FI & Molise & 28 &  &  \\
Sandra Savino & FI & Friuli-Venezia Giulia & 26 &  &  \\
Pasquale Pepe & Lega & Basilicata & 22 &  &  \\
Vincenza Bruno Bossio & PD & Calabria & 21 &  &  \\
Maria Rachele Ruiu & FdI & Lazio 1 & 20 &  &  \\
Francesca Cucchiara & AVS & Lombardia 1 & 12 &  &  \\
Francesco Michelotti & FdI & Toscana & 11 &  &  \\
Michele Marone & Lega & Molise & 9 &  &  \\
Simona Sassara & M5S & Lazio 2 & 9 &  &  \\
Elena Carnevali & PD & Lombardia 1 & 6 &  &  \\
Alessia Rotta & PD & Veneto 2 & 4 &  &  \\
Antonio Bramante & AVS & Puglia & 3 &  &  \\
Gregorio Fontana & FI & Veneto 2 & 3 &  &  \\
Claudia Porchietto & FI & Piemonte 1 & 2 &  &  \\
Luca Carabetta & M5S & Piemonte 1 & 1 &  &  \\
Mario Iacopino & M5S & Piemonte 2 & 1 &  &  \\
Mirella Cristina & FI & Piemonte 2 & 1 &  &  \\
\end{longtable}
\end{small}

\end{document}